\documentclass[journal,12pt,onecolumn,]{IEEEtran}
\IEEEoverridecommandlockouts
\usepackage{graphicx}
\usepackage{cite}
\usepackage{epstopdf}
\usepackage{amssymb, amsmath,amsthm, amsfonts}
\usepackage{ifthen}
\usepackage{tikz}
\usepackage{enumerate}
\usepackage{verbatim}
\usepackage{bbm,balance}
\usepackage{amsmath}
\usepackage{url}

\usetikzlibrary{arrows,shapes}
\usetikzlibrary{positioning}
\usetikzlibrary{calc}

\allowdisplaybreaks

\def\eE{\mathbb E}

\newcommand\independent{\protect\mathpalette{\protect\independenT}{\perp}}
\def\independenT#1#2{\mathrel{\rlap{$#1#2$}\mkern2mu{#1#2}}}

\newcommand*\fooA{\mathrel{-\mkern-3mu{\circ}\mkern-3mu-}}

\newtheorem{theorem}{Theorem}

\newtheorem{lemma}{Lemma}

\newtheorem{prop}{Proposition}
\newtheorem{defn}{Definition}
\theoremstyle{definition}\newtheorem{remark}{Remark}
\theoremstyle{definition}\newtheorem{example}{Example}

\newcommand{\be}{\begin{equation}}
\newcommand{\ee}{\end{equation}}
\newcommand{\ben}{\begin{equation*}}
\newcommand{\een}{\end{equation*}}
\newcommand{\ba}{\begin{eqnarray}}
\newcommand{\ea}{\end{eqnarray}}

\allowdisplaybreaks
\begin{document}
\title{Multi-terminal Strong Coordination subject to Secrecy Constraints}
\author{Viswanathan~Ramachandran, Tobias~J.~Oechtering and Mikael~Skoglund
\thanks{The authors are with the Division of Information Science and Engineering, School of Electrical Engineering and Computer Science, KTH Royal Institute of Technology, Stockholm, Sweden. (e-mails: \{visra,oech,skoglund\}@kth.se). Parts of this work have appeared in the proceedings of IEEE International Symposium on Information Theory (ISIT) 2024 as well as International Zurich Seminar (IZS) 2024.
}
}
\maketitle

\begin{abstract}
A fundamental problem in decentralized networked systems (such as for sensing or control) is to coordinate actions of different agents so that they reach a state of agreement. In such applications, it is additionally desirable that the actions at various nodes may not be anticipated by malicious eavesdroppers. Motivated by this, we investigate the problem of secure multi-terminal strong coordination aided by a multiple-access wiretap channel (MAC-WT). In this setup, independent and identically distributed (i.i.d.) copies of correlated sources are observed by two transmitters who encode the channel inputs to the MAC-WT. The legitimate receiver observing the channel output and side information correlated with the sources must produce approximately i.i.d. copies of an output variable jointly distributed with the sources. Furthermore, we demand that an external eavesdropper learns essentially nothing about the sources and the simulated output sequence by observing its own MAC-WT output. This setting is aided by the presence of independent pairwise shared randomness between each encoder and the legitimate decoder, that is unavailable to the eavesdropper. We derive an achievable rate region based on a combination of coordination coding and wiretap coding, along with an outer bound. The inner bound is shown to be tight and a complete characterization is derived for the special case when the sources are conditionally independent given the decoder side information and the legitimate channel is composed of deterministic links. Further, we also analyze a more general scenario with possible encoder cooperation, where one of the encoders can non-causally crib from the other encoder's input, for which an achievable rate region is proposed. We then explicitly compute the rate regions for an example both with and without cribbing between the encoders, and demonstrate that cribbing strictly improves upon the achievable rate region. 

\end{abstract}

\section{Introduction}
Distributed simulation of correlated randomness across different nodes in a network is useful in several settings, like cryptographic protocols and distributed computation tasks~\cite{ahlswede1993common,ahlswede1998common}. The role of communication to establish remote correlation in the form of a desired joint distribution of actions among all nodes in a network was explored in the \emph{coordination capacity} framework of \cite{cuff2010coordination}. This is especially relevant in scenarios where distributed agents must achieve decentralized cooperation (for instance in autonomous vehicle applications). Indeed, future telecommunication networks will provide services beyond the traditional role of communication of information, for which the cooperation and coordination of constituent devices will be crucial. Coordination coding can be used to model such services and obtain fundamental performance bounds that will serve as valuable benchmarks. In particular, we shall focus on the notion of \emph{strong coordination}~\cite{cuff2010coordination}, which requires that the distribution of the sequence of agent/node actions in the network be close in total variation to a target distribution. 
Furthermore, due to the trend towards large-scale decentralized networks consisting of many mutually distrusting terminals, security and integrity of the coordinated actions are also of high priority in order to guarantee trustworthy operation. In this work, we therefore address strong coordination over a (noisy) multi-user channel by taking information-theoretic secrecy aspects into account. The fundamental limits derived and the insights drawn in this framework could be broadly applicable towards a better understanding of the information flows and efficient resource utilization in interconnected systems, such as for smart grid or telemedicine. 

\begin{figure}[ht]
\centering
\begin{tikzpicture}[thick]
\node (d1) at (-3,0) [rectangle, draw, right, minimum height=1.0cm, minimum width = 1.2cm]{Enc $1$};
\node (d2) at (3.6,-0.95) [rectangle, draw, right, minimum height=2.5cm, minimum width = 1.6cm]{Dec};
\draw[->] (d2) --++(1.5,0) node[right]{$Y^n$};
\draw[<->] (d1) to[out=60,in=120] node[midway, above] {$K_1 \in [1:2^{nR_{01}}]$} (d2);
\draw[<-] (d1) --++(-1.5,0) node[left]{$X_1^n$};
\node (e2) at (-3,-2) [rectangle, draw, right, minimum height=1.0cm, minimum width = 1.2cm]{Enc $2$};
\draw[<-] (e2) --++(-1.5,0) node[left]{$X_2^n$};
\draw[<-] (d2.80) --++(0,1.5) |- (-2.95,3.1) node[left]{$W^n$};
\node (ch) at (-0.2,-1) [rectangle, draw, right, minimum height=0.8cm, minimum width = 1.2cm]{$p(\tilde{y},\tilde z|\tilde{x}_1,\tilde{x}_2)$};
\draw[->] (d1) -- (ch) node[midway, above, sloped]{$\tilde{X}_1^n$};
\draw[->] (e2) -- (ch) node[midway, above, sloped]{$\tilde{X}_2^n$};
\draw[->] (ch) -- (d2) node[midway, above]{$\tilde{Y}^n$};
\node (eve) at (0.45,-2.4) [rectangle, draw, right, minimum height=0.8cm, minimum width = 1.2cm]{Eve};
\draw[->] (ch) -- (eve) node[midway, right]{$\tilde Z^n$};
\draw[<->] (e2) to[out=-60,in=-120] node[midway, below] {$K_2 \in [1:2^{nR_{02}}]$} (d2);
\end{tikzpicture}
\caption{Strong coordination over a MAC-WT subject to secrecy constraints. Encoder $j\in \{1,2\} $ has access to the source $X_j^n$ and the shared randomness index $K_j$, and encodes the channel input $\tilde{X}_j^n$ for transmission over a discrete-memoryless MAC-WT $p(\tilde{y},\tilde z|\tilde{x}_1,\tilde{x}_2)$. The channel outputs $(\tilde{Y}^n,\tilde Z^n)$, where $\tilde{Y}^n$ is observed by the legitimate receiver, while $\tilde Z^n$ is observed at the eavesdropper. The legitimate receiver also observes side information $W^n$, where $(X_{1i},X_{2i},W_i)$, $i=1,\dots,n$, are assumed to be i.i.d. according to $q_{X_1X_2W}$, given by nature. The decoder must simulate an output sequence $Y^n$ such that $(X_{1i},X_{2i},W_i,Y_i)$, $i=1,\dots,n$, are approximately i.i.d. according to the distribution $q_{X_1X_2WY}$. Moreover, the coordinated actions $(X_1^n,X_2^n,W^n,Y^n)$ must be independent of $\tilde Z^n$, i.e., a strong secrecy constraint must hold against the eavesdropper.} \label{fig:encSRnoisy}
\end{figure}
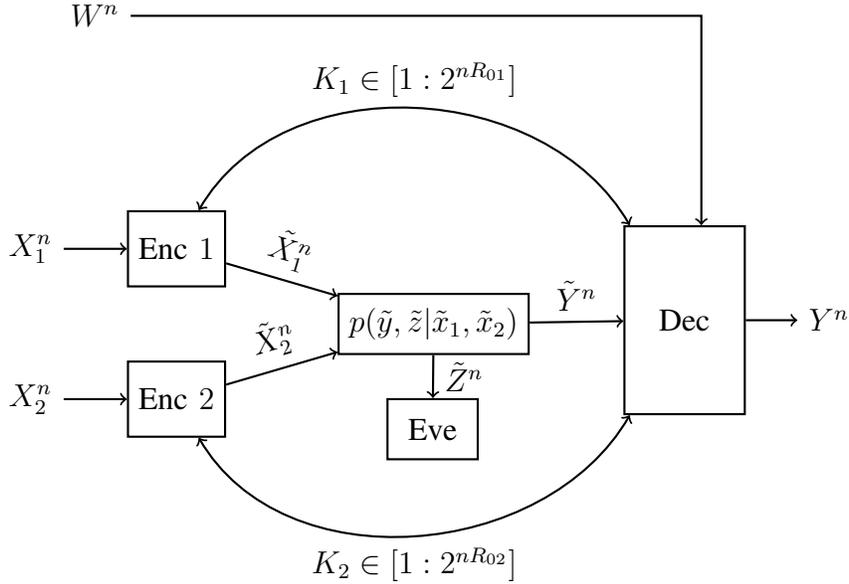

In this paper, we address secure distributed strong coordination via a MAC-WT. In particular, there are two encoders who encode the channel inputs to the MAC-WT, and a legitimate decoder who observes one of the outputs of the MAC-WT -- see Figure~\ref{fig:encSRnoisy}. There is also an external eavesdropper who observes the other output of the MAC-WT. The two encoders observe i.i.d. realizations of the sources $X_1$ and $X_2$ respectively, while the legitimate decoder observes i.i.d. realizations of the side information $W$, which are assumed to be drawn i.i.d. according to a given probability mass functions (p.m.f.) $q_{X_1X_2W}$. We will also separately analyze the case where the channel input $\tilde X_2^n$ is available non-causally to the first encoder, i.e., cribbing in the sense of~\cite[Situation 4]{willems1985discrete} is allowed (Section~\ref{UPDATEDsec:SMRnoisyK0}). Each encoder also has access to independent shared randomness with the legitimate decoder, which is independent of the sources and side information. Additionally, the encoders and the decoder are free to use (unconstrained) private randomization. The goal is to ensure that the decoder's output, along with the sources and side information, closely follows a target i.i.d. p.m.f. $q_{X_1X_2WY}$. Moreover, the secrecy requirement is such that the sources, side information and the simulated output sequence must appear to be independent of the eavesdropper channel observations. We discuss related prior work next, and explain the differences with respect to our model, before listing our key contributions.

Strong coordination aims to describe the minimal communication required to achieve remote correlation. An encoder observing an i.i.d. source $X^n$ with distribution $q_X$ transmits a message to a decoder through a noiseless link in the point-to-point formulation. The decoder's task is to produce a sequence $Y^n$ such that the total variation distance between the induced joint distribution of $(X^n,Y^n)$ and the i.i.d. joint distribution obtained by transmitting the source $X^n$ through a discrete memoryless channel $q_{Y|X}$ vanishes asymptotically with the blocklength. 
The encoder and the decoder may benefit from common randomness to accomplish this task. This paradigm has been studied under infinite common randomness~\cite{BennettSST02}, and a `reverse Shannon theorem' is established, where the minimum communication requirement is found to be $I(X;Y)$. Subsequently, the complete optimal trade-off region between communication and common randomness rates has been discovered in \cite{6757002} and \cite{cuff2013distributed}. Expanding on this, \cite{YassaeeGA15} obtained a complete characterization for a point-to-point network involving interactive communications between the nodes.

Building on the point-to-point network for strong coordination from \cite{cuff2013distributed}, multi-terminal extensions have received some attention in the literature. For instance, a cascade network is examined in \cite{SatpathyC16}, and the optimal trade-off between communication and common randomness rates is determined. Strong coordination over a multi-hop line network is addressed in \cite{vellambi2017strong}, using a channel-resolvability approach. A three-terminal generalization of \cite{cuff2013distributed} to include side information at the receiver is addressed in \cite{ramachandran2020strong}, where some tight characterizations are obtained for specific cases. Strong coordination over a multiple-access network of noiseless links is investigated in \cite{kurri2022multiple}, where a tight characterization is derived for independent sources. Further, the benefits of shared randomness amongst the encoders in reducing the communication needed for coordination is shown therein. Recently, a tight characterization is derived in \cite{RamachandranOSITW2024} for the same multiple-access coordination setup when the source observations are degraded. 

Strong coordination using noisy channels as a resource, rather than noiseless communication links as in \cite{cuff2013distributed}, was first explored in \cite{HaddadpourYAG13,HaddadpourYBGAA17}. For single-user as well as broadcast channel simulation, inner and outer bounds on the region of shared randomness rates were derived therein. A variant of \cite{HaddadpourYBGAA17} with an additional constraint of coordinating the channel input and output signals is addressed in \cite{CerviaLLB20}, where inner and outer bounds along with tight characterizations for specific cases (such as noiseless channels and lossless decoders) are derived. 

Strong coordination has also been explored in adversarial settings with secrecy constraints. For instance, a two-terminal noiseless setting with an external eavesdropper who taps into the communication between legitimate nodes and has access to correlated observations is investigated in \cite{gohari2012secure}, where an achievable scheme is derived. On a somewhat different note, strong coordination between two mutually distrusting users with no external adversaries is addressed in \cite{data2020interactive}, where the class of randomized functions which can be computed with perfect security is characterized. A point-to-point strong coordination setup over a noisy channel where the generated actions must be kept secret from an external eavesdropper is addressed in \cite{cervia2020secure}. The current work is a consolidated version of the conference papers \cite{RamachandranOSISIT2024}, which addressed an extension of \cite{kurri2022multiple} to noisy channels with secrecy constraints against an external eavesdropper, and \cite{RamachandranOSIZS2024}, which investigated a noiseless network version thereof.

The current paper extends channel simulation from noisy channels addressed in \cite{HaddadpourYBGAA17} to a three-terminal secrecy scenario with correlated sources, with the key difference being the presence of an external eavesdropper against whom strong secrecy is desired, apart from allowing for possible encoder cooperation (see Section~\ref{UPDATEDsec:SMRnoisyK0}). In other words, the correlated sources must be encoded over the MAC-WT so as to achieve strong coordination with the decoder output, while also ensuring secrecy against the eavesdropper. This can also be viewed as an extension of the multi-terminal noiseless network coordination problem addressed in \cite{kurri2022multiple} to the case of noisy resource channels, with additional secrecy constraints. Compared to \cite{gohari2012secure}, our setting explores secure strong coordination in a multi-terminal noisy channel setup rather than a point-to-point noiseless network with an eavesdropper. Our work also extends the single-user noisy channel secure coordination setting of \cite{cervia2020secure} to multi-user channels.

Ultimately, let us note that even with noiseless communication and/or in the absence of secrecy constraints, problems involving distributed encoding of correlated sources, such as the Berger-Tung source coding setup~\cite{berger1977multiterminal} or deterministic function computation settings such as~\cite{SefidgaranT11,SefidgaranT16} do not admit tight characterizations in general. For instance, \cite{Gastpar04} showed that for distributed lossy compression of correlated sources with decoder side information, a tight characterization results only when the sources are conditionally independent given the side information. Likewise, \cite{SefidgaranT11,SefidgaranT16} determined the optimal communication rates for computing deterministic functions (for which shared randomness does not help) in a noiseless multiple-access network only when the sources are conditionally independent given the side information. Otherwise, their inner bound is not tight in general, as illustrated by the K\"{o}rner-Marton problem~\cite{KornerM79} (see \cite[Example 2]{SefidgaranA11}). Since our coordination setup can also be viewed as distributed \emph{randomized} function computation, our problem appears even more challenging. Consequently, we also obtain tight results only for the special cases identified in the sequel, which is partly motivated by the assumptions required for tightness in the problems just mentioned.

\noindent \textbf{Main Contributions.}
Our main contributions are stated below.
\begin{itemize}
\item We derive an achievable region and an outer bound for the general case (refer Theorems~\ref{thm:encsideIBnoisy} and \ref{thm:encsideOBnoisy} in the sequel). Our achievable scheme is proved using a combination of coordination coding over noisy channels along with wiretap coding for secrecy, where the secrecy analysis is novel compared to prior works such as \cite{HaddadpourYAG13,HaddadpourYBGAA17}.
\item We also establish a complete characterization for this setting (see Theorem~\ref{thm:indepnoisy}) for the special case when the sources are conditionally independent given the decoder side information and the legitimate receiver's channel is composed of deterministic links. The non-trivial and challenging part lies in leveraging the deterministic legitimate channel and conditionally independent sources assumptions to obtain a single-letterization matching the inner bound, which is known to be difficult even in distributed source coding settings such as~\cite{berger1977multiterminal}.
\item Further, we analyze a more general scenario where one encoder is allowed to non-causally crib~\cite[Situation 4]{willems1985discrete} from the other encoder's input so as to enable encoder cooperation, for which an inner bound is proposed (see Theorem~\ref{thm:encsideIBnoisyK0}).
\item We then explicitly compute the regions for an example (Example~\ref{ex:1}) both with and without encoder cribbing and demonstrate that cribbing strictly improves upon the achievable region (see Section~\ref{sec:example}). The interesting part here lies in establishing the optimality of the auxiliary random variable choices for the setting without cribbing.
\end{itemize}

\textbf{Notations}: We denote random vectors as $A^n\triangleq (A_1, A_2 \cdots, A_n$), as is conventional in the information theory literature. Moreover, independence between two random variables $A$ and $B$ will be represented by $A \independent B$. Also, we use the notation $A \fooA B \fooA C$ for a Markov chain, in that $I(A;C|B)=0$.

The rest of this work is structured as follows. Section \ref{UPDATEDsec:SMRnoisy} introduces our system model, and Section \ref{secton:mainresults1noisy} presents the statements of our main results. 
The setting where the encoders are augmented with cribbing is presented in Section~\ref{UPDATEDsec:SMRnoisyK0}, while Section~\ref{sec:example} gives an example where such encoder cribbing helps compared to the optimal region without cribbing. The proofs of all the results are contained in Sections~\ref{app:pfThm1noisy} -- \ref{sec:exproof}. Finally, Section~\ref{sec:conc} concludes the paper with some closing remarks and future research avenues.

\section{System Model} \label{UPDATEDsec:SMRnoisy}
We investigate strong coordination of signals across three terminals over a MAC-WT with secrecy constraints. 
The setup comprises two encoders and a legitimate decoder which observe their respective source variables $(X_1^n,X_2^n,W^n)$, along with an eavesdropper. For $j \in \{1,2\}$, Encoder $j$ and the decoder can harness pairwise shared randomness $K_j$, assumed to be uniformly distributed on $[1:2^{nR_{0j}}]$. Encoder $j \in \{1,2\}$ (which observes $X_j^n$ and has access to $K_j$) encodes the channel input sequence $\tilde{X}_j^n$. A discrete-memoryless MAC-WT specified by $p(\tilde{y},\tilde z|\tilde{x}_1,\tilde{x}_2)$ maps the channel input sequences into an observation $\tilde{Y}^n$ at the legitimate receiver and $\tilde Z^n$ at the eavesdropper (who does not observe the shared randomness variables $(K_1,K_2)$). The variables $(X_{1i},X_{2i},W_i)$, $i=1,2,\ldots,n$, are assumed to be i.i.d. with joint distribution specified by nature as $q_{X_1X_2W}$. The random variables $X_1,X_2,W,\tilde{X}_1,\tilde{X}_2,\tilde{Y},\tilde Z$ assume values in finite alphabets $\mathcal{X}_1,\mathcal{X}_2,\mathcal{W},\mathcal{\tilde{X}}_1,\mathcal{\tilde{X}}_2,\mathcal{\tilde{Y}},\mathcal{\tilde{Z}}$, respectively. The shared randomness indices $K_1$ and $K_2$ are assumed to be independent of each other and of $(X_1^n,X_2^n,W^n)$. The legitimate decoder obtains $(K_1,K_2,W^n,\tilde{Y}^n)$ and simulates an output sequence $Y^n$ (where $Y_i$, $i=1,\dots,n,$ takes values in the finite alphabet $\mathcal{Y}$) which along with the input sources and side information must be approximately i.i.d. according to $q^{(n)}_{X_{1}X_{2}WY}(x_1^n,x_2^n,w^n,y^n):=\prod_{i=1}^n q_{X_1X_2WY}(x_{1i},x_{2i},w_i,y_i)$ (refer Figure~\ref{fig:encSRnoisy} for the details). Moreover, we require strong secrecy against the eavesdropper, i.e., the coordinated actions $(X_1^n,X_2^n,W^n,Y^n)$ must appear to be independent of the eavesdropper observations $\tilde Z^n$.

\begin{defn}\label{defn:codenoisy}
A $(2^{nR_{01}}, 2^{nR_{02}}, n)$ \emph{code} comprises two randomized encoders $p^{\emph{Enc}_j}(\tilde{x}_j^n|x_j^n,k_j)$ for $j \in \{1,2\}$ and a randomized decoder $p^{\emph{Dec}}(y^n|k_1,k_2,w^n,\tilde{y}^n)$, where $k_j\in[1:2^{nR_{0j}}]$, $j \in \{1,2\}$.
\end{defn}
The induced joint distribution of all random variables $(X_1^n,X_2^n,W^n,K_1,K_2,\tilde{X}_1^n,\tilde{X}_2^n,\tilde{Y}^n,\tilde Z^n,Y^n)$ and the resulting induced marginals on $(X_1^n,X_2^n,W^n,Y^n,\tilde Z^n)$ and $\tilde Z^n$ are respectively given by
\begin{align*}
p(x_1^n,x_2^n,w^n,k_1,k_2,\tilde{x}_1^n,\tilde{x}_2^n,\tilde{y}^n,\tilde z^n,y^n) &=\frac{1}{2^{n(R_{01}+R_{02})}}q(x_1^n,x_2^n,w^n)\prod_{j=1}^2 p^{\text{Enc}_j}(\tilde{x}_j^n|x_j^n,k_j) \\
& \hspace{24pt}\times p(\tilde{y}^n,\tilde z^n|\tilde{x}_1^n,\tilde{x}_2^n) p^{\text{Dec}}(y^n|k_1,k_2,w^n,\tilde{y}^n),
\end{align*}
and
\begin{align*}
&p^{\text{ind}}(x_1^n,x_2^n,w^n,y^n,\tilde z^n)=\sum_{k_1,k_2,\tilde{x}_1^n,\tilde{x}_2^n,\tilde{y}^n}p(x_1^n,x_2^n,w^n,k_1,k_2,\tilde{x}_1^n,\tilde{x}_2^n,\tilde{y}^n,\tilde z^n,y^n),\\
&p^{\text{ind}}(\tilde z^n)=\sum_{x_1^n,x_2^n,w^n,y^n}p^{\text{ind}}(x_1^n,x_2^n,w^n,y^n,\tilde z^n).
\end{align*}
The total variation distance between two p.m.f.'s, denoted as $p_X$ and $q_X$, defined on the same alphabet $\mathcal{X},$ is given by
\begin{align*}
||p_X-q_X||_1 \triangleq \frac{1}{2} \sum_{x\in\mathcal{X}} |p_X(x)-q_X(x)|.
\end{align*}

\begin{defn} \label{def:achnoisy}
A rate pair $(R_{01},R_{02})$ is said to be \emph{achievable for a target joint distribution} $q_{X_1X_2WY}$ \emph{with secrecy} provided there exists a sequence of $(2^{nR_{01}}, 2^{nR_{02}}, n)$ codes such that
\end{defn} 
\begin{align}\label{eqn:correctnessnoisy}
\lim_{n \to \infty} ||p^{\text{ind}}_{X_1^n,X_2^n,W^n,Y^n,\tilde Z^n}-p^{\text{ind}}_{\tilde Z^n} \cdot q^{(n)}_{X_1X_2WY}||_{1}=0,
\end{align}
where $q^{(n)}_{X_1X_2WY}$ is the target i.i.d. product p.m.f. defined as
\begin{align*}
&q^{(n)}_{X_1X_2WY}(x_1^n,x_2^n,w^n,y^n):=\prod_{i=1}^n q_{X_1X_2WY}(x_{1i},x_{2i},w_i,y_i).
\end{align*} 

The condition \eqref{eqn:correctnessnoisy} imposed on the joint distribution leads to strong secrecy against the eavesdropper provided that the total variation distance goes to zero exponentially in $n$, and we obtain
\begin{align}\label{eqn:secrecynoisy}
\lim_{n \to \infty} I(\tilde Z^n;X_1^n,X_2^n,W^n,Y^n)=0.
\end{align}


\begin{defn}\label{defn:newnoisy}
The \emph{rate region} $\mathcal{R}_{\textup{noisy-coord}}^{\textup{secrecy}}$ is the closure of the set of all achievable rate pairs $(R_{01},R_{02})$. 
\end{defn}
Let $\mathcal{R}_{\textup{noisy-coord, $R_{02} \to \infty$}}^{\textup{secrecy}}$ be the region when the pairwise shared randomness $K_2$ is unlimited, i.e., 
\begin{align}&\mathcal{R}_{\textup{noisy-coord, $R_{02} \to \infty$}}^{\textup{secrecy}}=\{R_{01} : \exists \ R_{02} \ \text{s.t.}\ (R_{01},R_{02}) \in \mathcal{R}_{\textup{noisy-coord}}^{\textup{secrecy}}\}.
\end{align}

\section{Main Results}  \label{secton:mainresults1noisy}
Firstly, let us present an inner bound to the rate region $\mathcal{R}_{\textup{noisy-coord}}^{\textup{secrecy}}$. In the following theorem, the underlying role played by the auxiliary random variables $U_{1}$ and $U_{2}$ is analogous to the auxiliary random variable in the point-to-point strong coordination setting over a noisy channel~\cite{HaddadpourYBGAA17}. Specifically, $U_{1}$ and $U_{2}$ are used to send source descriptions of $X_1$ and $X_2$ respectively, while $T$ is a time-sharing random variable. Furthermore, to ensure secrecy, we have two more auxiliary random variables $V_{1}$ and $V_{2}$ that create an implicit wiretap code and ensure the independence of the coordinated actions from the eavesdropper channel observations. The decoder then recovers the source descriptions and locally simulates $Y$ using a test channel (conditional distribution).
\begin{theorem}[Achievable Rate Region] \label{thm:encsideIBnoisy}
Given a target joint p.m.f. $q_{X_1X_2WY}$, the rate pair $(R_{01},R_{02})$ is in $\mathcal{R}_{\textup{noisy-coord}}^{\textup{secrecy}}$ provided
\begin{subequations}
\begin{align} 
I(V_{1};V_{2},\tilde{Y}|T) &\geq I(U_{1};X_1|U_{2},W,T) \label{eq:1a}\\
I(V_{2};V_{1},\tilde{Y}|T) &\geq I(U_{2};X_2|U_{1},W,T) \label{eq:1b}\\
I(V_{1},V_{2};\tilde{Y}|T) &\geq I(U_{1},U_{2};X_1,X_2|W,T) \label{eq:1c}\\
R_{01} &\geq I(U_{1};X_1,X_2,Y|W,T)\! -\!I(U_{1};U_{2}|W,T)+I(V_{1};\tilde Z|T)-I(V_{1};V_{2},\tilde{Y}|T) \label{eq:1d}\\
R_{02} &\geq I(U_{2};X_1,X_2,Y|W,T) \!-\!I(U_{1};U_{2}|W,T)+I(V_{2};\tilde Z|T)-I(V_{2};V_{1},\tilde{Y}|T) \label{eq:1e}\\
R_{01} &\geq I(U_{1};X_1,X_2,Y|W,T)+I(U_{2};X_2|U_{1},W,T)+I(V_{1};\tilde Z|T) \notag\\
&\hspace{12pt}-I(V_{1};\tilde{Y}|T)-I(V_{2};V_{1},\tilde{Y}|T) \label{eq:1f}\\
R_{02} &\geq I(U_{2};X_1,X_2,Y|W,T)+I(U_{1};X_1|U_{2},W,T)+I(V_{2};\tilde Z|T) \notag\\
&\hspace{12pt}-I(V_{2};\tilde{Y}|T)-I(V_{1};V_{2},\tilde{Y}|T) \label{eq:1g}\\
R_{01}+R_{02} &\geq I(U_{1},U_{2};X_1,X_2,Y|W,T)+I(V_{1},V_{2};\tilde Z|T)-I(V_{1},V_{2};\tilde{Y}|T) \label{eq:1h},
\end{align}
\end{subequations}
for some p.m.f. 
\begin{align}
p(x_1,x_2,w,t,u_{1},u_{2},v_{1},v_{2},\tilde x_1,\tilde x_2,\tilde y,\tilde z,y)&= p(x_1,x_2,w)p(t)\biggl(\prod_{j=1}^2 p(u_{j}|x_j,t)\biggr)p(v_{1}|t)p(\tilde x_1|v_{1},t) \notag\\
&\hspace{24pt} \times p(v_{2}|t)p(\tilde x_2|v_{2},t)p(\tilde y,\tilde z|\tilde x_1,\tilde x_2)p(y|u_{1},u_{2},w,t) \label{eq:pmfthm1noisy}
\end{align}
such that 
\begin{align*}
&\sum\limits_{u_{1},u_{2},v_{1},v_{2},\tilde x_1,\tilde x_2,\tilde y,\tilde z}  p(x_1,x_2,w,u_{1},u_{2},v_{1},v_{2},\tilde x_1,\tilde x_2,\tilde y,\tilde z,y|t) =q(x_1,x_2,w,y) \:\: \forall \:\: t.
\end{align*}
\end{theorem}

We note that constraints \eqref{eq:1a}--\eqref{eq:1c} ensure that the source descriptions can be successfully recovered at the decoder, while constraints \eqref{eq:1d}--\eqref{eq:1h} are the minimum rates of shared randomness needed for strong coordination with secrecy. A detailed proof of Theorem~\ref{thm:encsideIBnoisy} can be found in Section~\ref{app:pfThm1noisy}. We note that the independence between the source and channel variables in \eqref{eq:pmfthm1noisy} along with strong coordination of $(X_1^n,X_2^n,W^n,Y^n,\tilde Z^n)$ ensures the secrecy condition \eqref{eqn:secrecynoisy} against the eavesdropper.



\begin{remark} \label{rmk:noiselessMAC}
Consider the special case when the eavesdropper sees the same channel as the legitimate decoder, i.e. $\tilde Z=\tilde Y$. Further, let the memoryless MAC $p(\tilde y|\tilde x_1,\tilde x_2)$ consist of two independent point-to-point channels with capacities $R_1$ and $R_2$. Let $\tilde{X}_1$ (resp. $\tilde{X}_2$) have the capacity-achieving distribution for the first (resp. second) channel. In this case, taking $V_{1}=\tilde{X}_1$ and $V_{2}=\tilde{X}_2$ such that $\tilde X_1$ and $\tilde X_2$ are independent of $(U_1,U_2,X_1,X_2,W,Y)$, Theorem~\ref{thm:encsideIBnoisy} recovers the inner bound of \cite[Theorem 1]{RamachandranOSIZS2024}. The proof is detailed in Appendix~\ref{app:noiselessMAC}.
\end{remark}

We next present an outer bound to $\mathcal{R}_{\textup{noisy-coord}}^{\textup{secrecy}}$.
\begin{theorem}[Outer Bound] \label{thm:encsideOBnoisy}
Given a target p.m.f. $q_{X_1X_2WY}$, any rate pair $(R_{01},R_{02})$ in $\mathcal{R}_{\textup{noisy-coord}}^{\textup{secrecy}}$ satisfies, for every $\epsilon \in (0,\frac{1}{4}]$,
\begin{subequations}
\begin{align} 
I(\tilde{X}_1;\tilde{Y}|\tilde{X}_2,T) &\geq I(U_{1};X_1|U_{2},X_2,W,T) \label{eq:2a}\\
I(\tilde{X}_2;\tilde{Y}|\tilde{X}_1,T) &\geq I(U_{2};X_2|U_{1},X_1,W,T) \label{eq:2b}\\
I(\tilde{X}_1,\tilde{X}_2;\tilde{Y}|T) &\geq I(U_{1},U_{2};X_1,X_2|W,T) \label{eq:2c}\\
R_{01}+R_{02} &\geq I(U_{1},U_{2};X_1,X_2,Y|W,T)-H(Y|X_1,X_2,W,T)-I(\tilde{X}_1,\tilde{X}_2;\tilde{Y}|\tilde Z,T)-2g(\epsilon) \label{eq:2d}\\
R_{01} &\geq I(U_{1};X_1,X_2,Y|W,T)-H(Y|X_1,X_2,W,T)-I(\tilde{X}_1,\tilde{X}_2;\tilde{Y}|\tilde Z,T)-2g(\epsilon) \label{eq:2e}\\
R_{02} &\geq I(U_{2};X_1,X_2,Y|W,T)-H(Y|X_1,X_2,W,T)-I(\tilde{X}_1,\tilde{X}_2;\tilde{Y}|\tilde Z,T)-2g(\epsilon) \label{eq:2f},
\end{align}
\end{subequations}
with 
\begin{align*}
g(\epsilon)&=2\sqrt{\epsilon}\biggl(H_q(X_1,X_2,W,Y)+\log\frac{(|\mathcal{X}_1||\mathcal{X}_2||\mathcal{W}||\mathcal{Y}||\mathcal{\tilde Z}|)}{\epsilon}\biggr)
\end{align*}
(where $g(\epsilon) \to 0$ as $\epsilon \to 0$), for some p.m.f.  
\begin{align}
p(x_1,x_2,w,t,u_{1},u_{2},\tilde x_1,\tilde x_2,\tilde y,\tilde z,y)&= p(x_1,x_2,w)p(t)p(u_{1},u_{2}|x_1,x_2,t)p(\tilde x_1,\tilde x_2|t) \notag\\
&\hspace{18pt} \times p(\tilde y,\tilde z|\tilde x_1,\tilde x_2)p(y|u_{1},u_{2},w,t) \label{pmf:ob1noisy}
\end{align}
such that 
\begin{align}
&||p(x_1,x_2,w,y|t)-q(x_1,x_2,w,y)||_1 \leq \epsilon \: \textup{for all} \: t \notag.
\end{align} 
\end{theorem}
The main difference compared to the inner bound lies in the more general p.m.f. structure of the auxiliary random variables $(U_{1},U_{2})$ in relation to the source variables $(X_1,X_2)$, while a long Markov chain held in the inner bound due to distributed encoding. The proof details are given in Section~\ref{proof:OB-thm6}. In other words, the outer bound admits cooperation between the two auxiliary codebooks, while this is not achievable in general in the inner bound due to the distributed processing. This motivates the study of cribbing encoders in Section~\ref{UPDATEDsec:SMRnoisyK0}.

When the random variables $X_1$ and $X_2$ are conditionally independent given $W$, and the legitimate receiver's channel $p(\tilde{y}|\tilde{x}_1,\tilde{x}_2)$ is composed of deterministic links, i.e., $\tilde{Y}=(f_1(\tilde{X}_1),f_2(\tilde{X}_2))$ for deterministic maps $f_1(\cdot)$ and $f_2(\cdot)$, we can demonstrate the tightness of the inner bound in Theorem~\ref{thm:encsideIBnoisy} by obtaining a matching converse bound.
\begin{theorem}[Tight Characterization - Conditionally Independent Sources and Deterministic Legitimate Channel] \label{thm:indepnoisy}
Consider a target p.m.f. $q_{X_1X_2WY}$ such that $I(X_1;X_2|W)=0$, and also assume a legitimate channel for which $\tilde{Y}=(f_1(\tilde{X}_1),f_2(\tilde{X}_2)) \triangleq (\tilde Y_1,\tilde Y_2)$. Then $\mathcal{R}_{\textup{noisy-coord, $R_{02} \to \infty$}}^{\textup{secrecy}}$ is characterized by the set of rates $R_{01}$ such that 
\begin{subequations}
\begin{align} 
H(\tilde{Y}_1|T) &\geq I(U_{1};X_1|W,T) \label{eq:3a}\\
H(\tilde{Y}_2|T) &\geq I(U_{2};X_2|W,T) \label{eq:3b}\\
R_{01} &\geq I(U_{1};X_1,Y|X_2,W,T)-H(\tilde Y_1|\tilde Z,T) \label{eq:3c},
\end{align}
\end{subequations}
for some p.m.f. 
\begin{align}
p(x_1,x_2,w,t,u_{1},u_{2},\tilde x_1,\tilde x_2,\tilde z,y)&=p(w)p(x_1|w)p(x_2|w)p(t) \biggl(\prod_{j=1}^2 p(u_{j}|x_j,t)\biggr) \nonumber\\
&\hspace{24pt}\times p(\tilde x_1|t) p(\tilde x_2|t) p(\tilde z|\tilde x_1,\tilde x_2)p(y|u_{1},u_{2},w,t) \label{pmf:ob3noisy}
\end{align}
such that 
\begin{align*}
&\sum\limits_{u_{1},u_{2},\tilde x_1,\tilde x_2,\tilde z} p(x_1,x_2,w,u_{1},u_{2},\tilde x_1,\tilde x_2,\tilde z,y|t) =q(x_1,x_2,w,y),
\end{align*}
for all $t$, with $|\mathcal{U}_{1}| \leq |\mathcal{X}_1||\mathcal{X}_2||\mathcal{W}||\mathcal{Y}|$, $|\mathcal{U}_{2}| \leq |\mathcal{U}_{1}||\mathcal{X}_1||\mathcal{X}_2||\mathcal{W}||\mathcal{Y}|$ and $|\mathcal{T}| \leq 3$.
\end{theorem}
The achievability follows from Theorem~\ref{thm:encsideIBnoisy} by enforcing the conditional independence $p(x_1,x_2,w) = p(w)p(x_1|w)p(x_2|w)$ and the deterministic condition on the legitimate receiver's channel $\tilde{Y}=(f_1(\tilde{X}_1),f_2(\tilde{X}_2))$, along with sufficiently large shared randomness rate $R_{02}$. We also make the choices $V_{1}=\tilde{Y}_1$ and $V_{2}=\tilde{Y}_2$ in Theorem~\ref{thm:encsideIBnoisy}. A detailed proof of the converse is given in Section~\ref{proof:conv-thm5}. We note that the converse for this case does not follow from the outer bound in Theorem~\ref{thm:encsideOBnoisy}, and needs a stronger outer bound which matches the p.m.f. structure in the inner bound.

\begin{figure}[ht]
\centering
\begin{tikzpicture}[thick]
\node (d1) at (-3,0) [rectangle, draw, right, minimum height=1.0cm, minimum width = 1.2cm]{Enc $1$};
\node (d2) at (3.6,-0.95) [rectangle, draw, right, minimum height=2.5cm, minimum width = 1.6cm]{Dec};
\draw[->] (d2) --++(1.5,0) node[right]{$Y^n$};
\draw[<->] (d1) to[out=60,in=120] node[midway, above] {$K_1 \in [1:2^{nR_{01}}]$} (d2);
\draw[<-] (d1) --++(-1.5,0) node[left]{$X_1^n$};
\node (e2) at (-3,-2) [rectangle, draw, right, minimum height=1.0cm, minimum width = 1.2cm]{Enc $2$};
\draw[<-] (e2) --++(-1.5,0) node[left]{$X_2^n$};
\draw[<-] (d2.80) --++(0,1.5) |- (-2.95,3.1) node[left]{$W^n$};
\node (ch) at (-0.2,-1) [rectangle, draw, right, minimum height=0.8cm, minimum width = 1.2cm]{$p(\tilde{y},\tilde z|\tilde{x}_1,\tilde{x}_2)$};
\draw[->] (d1) -- (ch) node[midway, above, sloped]{$\tilde{X}_1^n$};
\draw[->] (e2) -- (ch) node[midway, below, sloped]{$\tilde{X}_2^n$};
\draw[->] (ch) -- (d2) node[midway, above]{$\tilde{Y}^n$};
\node (eve) at (0.45,-2.4) [rectangle, draw, right, minimum height=0.8cm, minimum width = 1.2cm]{Eve};
\draw[->] (ch) -- (eve) node[midway, right]{$\tilde Z^n$};
\draw[<->] (e2) to[out=-60,in=-120] node[midway, below] {$K_2 \in [1:2^{nR_{02}}]$} (d2);
\node (j) at (-1.2,-1.8) [rectangle]{};
\draw[->] (j) |- (-1.2,-1.2) to (-2.35,-1.2) -- (d1.south);
\end{tikzpicture}
\caption{Strong coordination over a MAC-WT subject to secrecy constraints, with encoder cribbing. The input $\tilde X_2^n$ is non-causally available to the first encoder so that encoder cooperation is facilitated.} \label{fig:encSRnoisycrib}
\end{figure}
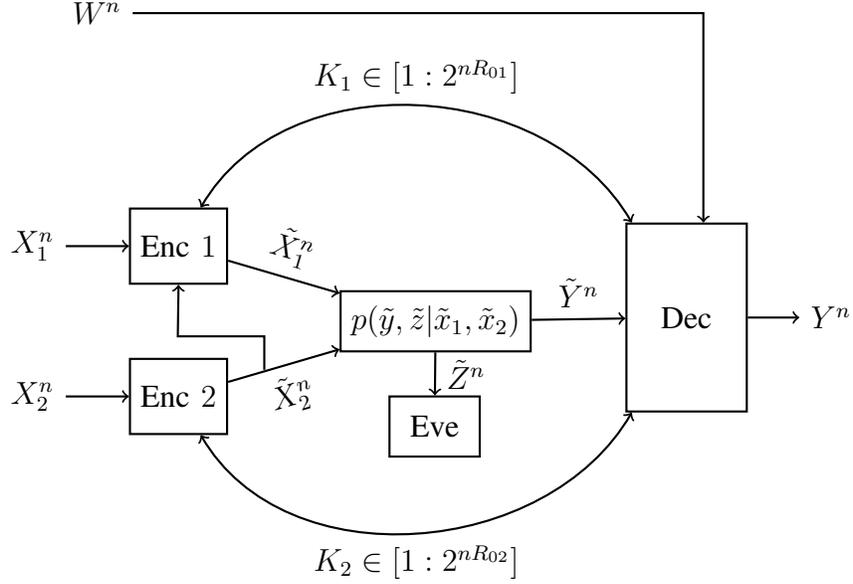 
\section{Secure Strong Coordination with Encoder Cribbing} \label{UPDATEDsec:SMRnoisyK0}
In this section, we present our results in the broader context where one of the encoders is allowed to crib~\cite[Situation 4]{willems1985discrete} from the other encoder's input non-causally -- see Fig.~\ref{fig:encSRnoisycrib}. This will facilitate cooperation between the encoders, in that Enc~$1$ can build its codebooks conditioned on the knowledge of the codebooks from Enc~$2$. In particular, Enc $2$ (which observes $X_2^n$ and has access to $K_2$) first generates the channel input sequence $\tilde{X}_2^n$. Then, Enc $1$ (which observes $X_1^n$ and has access to $K_1$ as well as $\tilde X_2^n$) creates the channel input sequence $\tilde{X}_1^n$.
A code, an achievable rate pair, and the rate region can be defined analogously as before. In particular, the code and an achievable rate pair can be defined similar to Definitions~\ref{defn:codenoisy} and \ref{def:achnoisy} by changing the map at Encoder~$1$ to $p^{\text{Enc}_1}(\tilde{x}_1^n|x_1^n,k_1,\tilde x_2^n)$. The \emph{rate region} $\mathcal{R}_{\textup{noisy-coord}}^{\textup{secrecy, crib}}$ is the closure of the set of all achievable rate tuples $(R_{01},R_{02})$ in the presence of cribbing.

We now present an inner bound to the region $\mathcal{R}_{\textup{noisy-coord}}^{\textup{secrecy, crib}}$.
\begin{theorem}[Achievable Rate Region with Cribbing amongst the Encoders] \label{thm:encsideIBnoisyK0}
Given a target joint p.m.f. $q_{X_1X_2WY}$, the rate pair $(R_{01},R_{02})$ is in $\mathcal{R}_{\textup{noisy-coord}}^{\textup{secrecy, crib}}$ provided
\begin{subequations}
\begin{align} 
I(U_{1};X_1|U_{2},W,T) &\leq I(V_{1};V_{2},\tilde{Y}|T)-I(V_1;\tilde X_2|T) \label{eq:4a}\\
I(U_{2};X_2|T)-I(U_2;U_1,W|T) &\leq I(V_{2};V_{1},\tilde{Y}|T) \label{eq:4b}\\
I(U_{1};X_1|U_{2},W,T)+I(U_{2};X_2|W,T) &\leq I(V_{1};V_{2},\tilde{Y}|T)+I(V_{2};\tilde{Y}|T) -I(V_1;\tilde X_2|T) \label{eq:4c}\\
R_{01} &\geq I(U_{1};X_1,X_2,Y|W,T)-I(U_{1};U_{2}|W,T) \notag\\
&\hspace{12pt}+I(V_{1};\tilde Z|T)-I(V_{1};V_{2},\tilde{Y}|T) \label{eq:4d}\\
R_{02} &\geq I(U_{2};X_1,X_2,Y|W,T)-I(U_{1};U_{2}|W,T) \notag\\
&\hspace{12pt}+I(V_{2};\tilde Z|T)-I(V_{2};V_{1},\tilde{Y}|T) \label{eq:4e}\\
R_{01} &\geq I(U_{1};X_1,X_2,Y|W,T)+I(U_{2};X_2|T) \notag\\
&\hspace{12pt}-I(U_2;U_1,W|T)+I(V_{1};\tilde Z|T)-I(V_{1};\tilde{Y}|T) \notag\\
&\hspace{12pt}-I(V_{2};V_{1},\tilde{Y}|T) \label{eq:4f}\\
R_{02} &\geq I(U_{2};X_1,X_2,Y|W,T)+I(U_{1};X_1|U_{2},W,T) \notag\\
&\hspace{12pt}+I(V_{2};\tilde Z|T)-I(V_{2};\tilde{Y}|T) \notag\\
&\hspace{12pt}+I(V_1;\tilde X_2|T)-I(V_{1};V_{2},\tilde{Y}|T) \label{eq:4g}\\
R_{01}+R_{02} &\geq I(U_{1},U_{2};X_1,X_2,Y|W,T) \notag\\
&\hspace{12pt}+I(V_{1},V_{2};\tilde Z|T)-I(V_{1},V_{2};\tilde{Y}|T) \label{eq:4h},
\end{align}
\end{subequations}
for some p.m.f. 
\begin{align}
p(x_1,x_2,w,t,u_{1},u_{2},v_{1},v_{2},\tilde x_1,\tilde x_2,\tilde y,\tilde z,y)&=p(x_1,x_2,w)p(t) p(u_{2}|x_2,t)p(u_{1}|x_1,u_2,t) \notag\\
&\hspace{24pt} \times p(v_{2}|t)p(\tilde x_2|v_{2},t) p(v_{1}|\tilde x_2,t)p(\tilde x_1|v_{1},\tilde x_2,t)  \notag\\
&\hspace{24pt} \times p(\tilde y,\tilde z|\tilde x_1,\tilde x_2)p(y|u_{1},u_{2},w,t) \label{eq:pmfthm1noisycrib}
\end{align}
such that 
\begin{align*}
&\sum\limits_{u_{1},u_{2},v_{1},v_{2},\tilde x_1,\tilde x_2,\tilde y,\tilde z}  p(x_1,x_2,w,u_{1},u_{2},v_{1},v_{2},\tilde x_1,\tilde x_2,\tilde y,\tilde z,y|t)  =q(x_1,x_2,w,y) \:\: \forall \:\: t.
\end{align*}
\end{theorem}

The key difference compared to Theorem~\ref{thm:encsideIBnoisy} is that Enc~$1$ can now build its codebooks conditioned on the knowledge of the codebooks from Enc~$2$, and this facilitates encoder cooperation. In particular, note that in the right-hand sides of the inequalities \eqref{eq:4a}, \eqref{eq:4c}, and \eqref{eq:4g}, the mutual information terms associated with the first user can now depend upon $\tilde X_2$. For a detailed proof, please refer to Section~\ref{proof:thm4}.


\section{Example: Cribbing Helps for Secure Channel Simulation} \label{sec:example}
In this section, we show with the help of an example that in the presence of cribbing between the encoders, the achievable region can be improved upon. The setting with encoder cribbing considered here resembles rate-distortion for correlated sources with partially separated encoders~\cite{kaspi1982rate}, wherein one of the encoders is supplied with partial information about the other encoder's source observation. Our illustration will be in the context of Theorem~\ref{thm:indepnoisy}, where we obtained a closed characterization for conditionally independent sources and deterministic legitimate channels (without encoder cribbing). Accordingly, we first compute the region of Theorem~\ref{thm:indepnoisy} (without encoder cribbing) for this example. We then show that the region is improved in the presence of encoder cribbing.  
\begin{example} \label{ex:1}
Let $X_1 \in \{0,1\}$ and $X_2 \in \{0,1\}$ be independent and uniform binary random variables. Suppose the channel to be simulated $q_{Y|X_1,X_2}$ is such that 
\begin{align}
Y = \begin{cases}
X_1 &\text{if $X_2=1$}\\
\text{`e'} &\text{if $X_2=0$},
\end{cases}
\end{align}
where the symbol `e' stands for an erasure.
For the sake of simplicity, we let $W=\varnothing$ and assume a perfect legitimate channel $\tilde Y=(\tilde X_1,\tilde X_2)$, i.e., the maps $f_1(\cdot)$ and $f_2(\cdot)$ in Theorem~\ref{thm:indepnoisy} are identities. Suppose that the eavesdropper's channel is specified by $\tilde Z=(\tilde X_1,\tilde X_2 \oplus N)$, where $N \sim \textup{Bern}(p)$, $p \in [0,\frac{1}{2}]$ and $N \independent (\tilde X_1,\tilde X_2)$. We focus on the required values of $H(\tilde X_1,\tilde X_2)$ and $R_{01}$ (under $R_{02} \to \infty$ as in Theorem~\ref{thm:indepnoisy}) for secure channel simulation, with and without encoder cribbing.   
\end{example}

In the absence of cribbing between the encoders, from Theorem~\ref{thm:indepnoisy}, the region $\mathcal{R}_{\textup{noisy-coord, $R_{02} \to \infty$}}^{\textup{secrecy}}$ for a perfect legitimate channel $\tilde Y_j=\tilde X_j$ for $j \in \{1,2\}$ and independent sources is simply characterized by the constraints
\begin{align*} 
H(\tilde{X}_1|T) &\geq I(U_{1};X_1|T) \\
H(\tilde{X}_2|T) &\geq I(U_{2};X_2|T) \\
R_{01} &\geq I(U_1;X_1,Y|X_2,T)-H(\tilde{X}_1|\tilde Z,T),
\end{align*}
for some p.m.f. 
\begin{align}
p(x_1,x_2,t,u_{1},u_{2},\tilde{x}_1,\tilde{x}_2,\tilde z,y)&=p(x_1)p(x_2)p(t) \biggl(\prod_{j=1}^2 p(u_{j}|x_j,t) p(\tilde x_j|t)\biggr) \nonumber\\
&\hspace{12pt} \times p(\tilde z|\tilde x_1,\tilde x_2) p(y|u_{1},u_{2},t)  
\end{align}
such that 
\begin{align*}
&\sum\limits_{u_{1},u_{2},\tilde{x}_1,\tilde{x}_2,\tilde z}p(x_1,x_2,u_{1},u_{2},\tilde{x}_1,\tilde{x}_2,\tilde z,y|t)=q(x_1,x_2,y),
\end{align*}
for all $t$. The following proposition explicitly characterizes the optimal region for the given $q_{X_1 X_2 Y}$, with the non-trivial part being the proof of converse.

\begin{prop} \label{prop:1}
For the target distribution $q_{X_1 X_2 Y}$ in Example~\ref{ex:1}, secure channel simulation is feasible if and only if the constraints $H(\tilde X_1) \geq 1$, $H(\tilde X_2) \geq 1$ (thus $H(\tilde X_1,\tilde X_2) \geq 2$), and $R_{01} \geq 1$ hold.
\end{prop}
\begin{proof}
For the achievability, it suffices to prove that there exists a p.m.f.
\begin{align*}
p(x_1,x_2,u_{1},u_{2},\tilde{x}_1,\tilde{x}_2,\tilde z,y)&=p(x_1)p(x_2) \biggl(\prod_{j=1}^2 p(u_{j}|x_j) p(\tilde x_j)\biggr) p(\tilde z|\tilde x_1,\tilde x_2) p(y|u_{1},u_{2})   
\end{align*}
such that $\sum\limits_{u_{1},u_{2},\tilde{x}_1,\tilde{x}_2,\tilde z}p(x_1,x_2,u_{1},u_{2},\tilde{x}_1,\tilde{x}_2,\tilde z,y)=q(x_1,x_2,y)$, along with $I(U_{1};X_1)=1$, $I(U_{2};X_2)=1$ and $I(U_1;X_1,Y|X_2)-H(\tilde{X}_1|\tilde Z)=1$. By choosing $U_1=X_1$ and $U_2=X_2$, it is clear that the conditions on the joint p.m.f. $p(x_1,x_2,u_{1},u_{2},\tilde{x}_1,\tilde{x}_2,\tilde z,y)$ are satisfied and $I(U_{1};X_1)=H(X_1)=1$, $I(U_{2};X_2)=H(X_2)=1$, along with $I(U_1;X_1,Y|X_2)-H(\tilde{X}_1|\tilde Z)=H(X_1|X_2)-H(\tilde{X}_1|\tilde{X}_1,\tilde{X}_2 \oplus N)=1$. The interesting part is the proof of converse, which is detailed in Section~\ref{sec:exproof}.
\end{proof}

We next prove that with cribbing, it is possible to 
achieve secure channel simulation with smaller values of $H(\tilde X_1,\tilde X_2)$ and $R_{01}$ compared to the optimal region without cribbing. 
The region $\mathcal{R}_{\textup{noisy-coord}}^{\textup{secrecy, crib}}$ in Theorem~\ref{thm:encsideIBnoisyK0} specializes for independent sources, perfect legitimate channel and unlimited shared randomness rate $R_{02}$ to the set of constraints (with $V_1=\tilde X_1$ and $V_2=\tilde X_2$)
\begin{align*}
H(\tilde{X}_1|\tilde X_2,T) &\geq I(U_{1};X_1|U_2,T) \\
H(\tilde{X}_2|T) &\geq I(U_{2};X_2|T)-I(U_{1};U_2|T)\\
H(\tilde{X}_1,\tilde{X}_2|T) &\geq I(U_{1};X_1|U_2,T)+I(U_{2};X_2|T)\\
R_{01} &\geq I(U_1;X_1,X_2,Y|T)-I(U_{1};U_2|T)-H(\tilde{X}_1|\tilde Z,T),
\end{align*}
for some p.m.f. 
\begin{align}
p(x_1,x_2,t,u_{1},u_{2},\tilde{x}_1,\tilde{x}_2,\tilde z,y)&=p(x_1)p(x_2)p(t) p(u_{2}|x_2,t)p(u_{1}|x_1,u_2,t)p(\tilde x_2|t)\nonumber\\
&\hspace{18pt} \times p(\tilde x_1|\tilde x_2,t)p(\tilde z|\tilde x_1,\tilde x_2) p(y|u_{1},u_{2},t)  
\end{align}
such that 
\begin{align*}
&\sum\limits_{u_{1},u_{2},\tilde{x}_1,\tilde{x}_2,\tilde z}p(x_1,x_2,u_{1},u_{2},\tilde{x}_1,\tilde{x}_2,\tilde z,y|t)=q(x_1,x_2,y),
\end{align*}
for all $t$. Now we can choose $U_2=X_2$ and
\begin{align}
U_1 = \begin{cases}
X_1 &\text{if $U_2=1$}\\
0 &\text{if $U_2=0$},
\end{cases}
\end{align}
to obtain that secure channel simulation is feasible if 
$H(\tilde X_1,\tilde X_2) \geq I(U_{1};X_1|U_2)+I(U_{2};X_2)=I(U_{1};X_1|X_2)+H(X_2)=0.5 \: I(U_1;X_1|X_2=0)+0.5 \: I(U_1;X_1|X_2=1)+1=0+0.5 H(X_1)+1=1.5$, along with $R_{01} \geq I(U_1;X_1,X_2,Y)-I(U_1;U_2)-H(\tilde{X}_1|\tilde Z)=I(U_1;X_1,X_2)+I(U_1;Y|X_1,X_2)-I(U_1;X_2)-H(\tilde{X}_1|\tilde{X}_1,\tilde{X}_2 \oplus N)=I(U_1;X_1|X_2)+0-0=0.5$. This strictly improves upon the (optimal) requirements without cribbing, which were $H(\tilde X_1,\tilde X_2) \geq 2$, along with $R_{01} \geq 1$. We thus conclude that secure channel simulation can be achieved with lesser values of both $H(\tilde X_1,\tilde X_2)$ as well as shared randomness rate $R_{01}$ in the presence of cribbing.

\section{Proof of Theorem~\ref{thm:encsideIBnoisy}} \label{app:pfThm1noisy}
The proof makes use of the Output Statistics of Random Binning (OSRB) framework of \cite{yassaee2014achievability}. In the following discussion, we adopt the convention of using capital letters (such as $P_{X}$) to represent random p.m.f.'s, as in \cite{cuff2013distributed,yassaee2014achievability}. On the other hand, lowercase letters (like $p_X$) denote non-random p.m.f.'s. We denote the uniform distribution over the set $\mathcal{A}$ as $p_{\mathcal{A}}^{\text{U}}$. The symbol $\approx$ is used for approximations of p.m.f.'s, following the convention in \cite{yassaee2014achievability}. For two random p.m.f.'s $P_X$ and $Q_X$ defined on the same alphabet $\mathcal{X}$, the notation $P_X \stackrel{\epsilon}\approx Q_X$ is used to indicate $\eE{[||P_X-Q_X||_1]} \leq \epsilon$. We define the notation $P_{X^n} \approx Q_{X^n}$ for any two sequences of random p.m.f.'s $P_{X^n}$ and $Q_{X^n}$ on $\mathcal{X}^n$ to indicate $\lim_{n \to \infty} \mathbb{E}[||P_{X^n}-Q_{X^n}||_1] = 0$. 

We establish achievability when $|\mathcal{T}|=1$, while the general case follows by time-sharing. We next define two protocols, one each based on random binning and random coding.\\
\underline{Random Binning Protocol:} Let the random variables $(U_{1}^n,V_{1}^n,U_{2}^n,V_{2}^n,X_1^n,X_2^n,W^n,\tilde{X}_1^n,\tilde{X}_2^n,\tilde{Y}^n,\tilde Z^n,Y^n)$ be drawn i.i.d. according to the joint distribution
\begin{align}
p(x_1,x_2,w,u_{1},v_{1},u_{2},v_{2},\tilde{x}_1,\tilde{x}_2,\tilde{y},\tilde z,y) &=p(x_1,x_2,w)p(u_{1},v_1|x_1)p(u_{2},v_2|x_2)p(y|u_{1},u_{2},w) \notag\\
&\phantom{ww}\times p(\tilde{x}_1|v_{1})p(\tilde{x}_2|v_{2})p(\tilde{y},\tilde z|\tilde{x}_1,\tilde{x}_2)
\end{align}
such that the marginal $p(x_1,x_2,w,y)=q(x_1,x_2,w,y)$ and the independence between the variables $(U_{1},U_{2},X_1,X_2,W,Y)$ and $ (V_{1},V_{2},\tilde{X}_1,\tilde{X}_2,\tilde{Y},\tilde Z)$ holds. 
The following random binning is then applied:
\begin{itemize}
\item Independently generate two uniform bin indices $(K_1,F_1)$ of $(U_{1}^n,V_{1}^n)$, where $K_1 = \phi_{1}(U_{1}^n,V_{1}^n) \in [1:2^{n R_{01}}]$ and $F_1 = \phi_{2}(U_{1}^n,V_{1}^n) \in [1:2^{n \tilde{R}_1}]$. 
\item Similarly, independently generate two uniform bin indices $(K_2,F_2)$ from $(U_{2}^n,V_{2}^n)$, where $K_2 = \phi_{3}(U_{2}^n,V_{2}^n) \in [1:2^{n R_{02}}]$ and $F_2 = \phi_{4}(U_{2}^n,V_{2}^n) \in [1:2^{n \tilde{R}_2}]$. Note that $F_1,F_2$ represents additional shared randomness assumed in the OSRB framework, which will be eliminated later without affecting the i.i.d. distribution. 
\end{itemize}
The receiver estimates $(\hat{u}_{1}^n,\hat{u}_{2}^n)$ from its observations $(k_1,f_1,k_2,f_2,w^n,\tilde{y}^n)$ using a Slepian-Wolf decoder.
The random p.m.f. induced by this binning scheme is:
\begin{align}
&P(x_1^n,x_2^n,w^n,y^n,u_{1}^n,u_{2}^n,v_{1}^n,v_{2}^n,\tilde{x}_1^n,\tilde{x}_2^n\tilde{y}^n,\tilde z^n,k_1,f_1,k_2,f_2,\hat{u}_{1}^n,\hat{u}_{2}^n) \notag\\
&= p(x_1^n,x_2^n,w^n)p(u_{1}^n,v_1^n|x_1^n)p(u_{2}^n,v_2^n|x_2^n)p(y^n|u_{1}^n,u_{2}^n,w^n) p(\tilde{x}_1^n|v_{1}^n)p(\tilde{x}_2^n|v_{2}^n)\notag\\
&\phantom{w} \times p(\tilde{y}^n,\tilde z^n|\tilde{x}_1^n,\tilde{x}_2^n)P(k_1,f_1|u_{1}^n,v_{1}^n) P(k_2,f_2|u_{2}^n,v_{2}^n)P^{SW}(\hat{u}_{1}^n,\hat{u}_{2}^n|k_1,f_1,k_2,f_2,w^n,\tilde{y}^n) \label{eq:osrb2noisy} \\
&= p(x_1^n,x_2^n,w^n) P(k_1,f_1,u_{1}^n,v_{1}^n|x_1^n) P(k_2,f_2,u_{2}^n,v_{2}^n|x_2^n)p(\tilde{x}_1^n|v_{1}^n)p(\tilde{x}_2^n|v_{2}^n) \notag\\
&\phantom{w} \times p(\tilde{y}^n,\tilde z^n|\tilde{x}_1^n,\tilde{x}_2^n) P^{SW}(\hat{u}_{1}^n,\hat{u}_{2}^n|k_1,f_1,k_2,f_2,w^n,\tilde{y}^n) p(y^n|u_{1}^n,u_{2}^n,w^n) \notag\\
&= p(x_1^n,x_2^n,w^n) P(k_1,f_1|x_1^n) P(u_{1}^n,v_{1}^n|k_1,f_1,x_1^n)P(k_2,f_2|x_2^n) P(u_{2}^n,v_{2}^n|k_2,f_2,x_2^n) \notag\\
&\phantom{w} \times p(\tilde{x}_1^n|v_{1}^n)p(\tilde{x}_2^n|v_{2}^n)p(\tilde{y}^n,\tilde z^n|\tilde{x}_1^n,\tilde{x}_2^n) P^{SW}(\hat{u}_{1}^n,\hat{u}_{2}^n|k_1,f_1,k_2,f_2,w^n,\tilde{y}^n) p(y^n|u_{1}^n,u_{2}^n,w^n). \notag
\end{align}

\noindent \underline{Random Coding Protocol:} In this scheme, we assume the presence of additional shared randomness $F_j$ of rate $\tilde{R}_j, j \in \{1,2\}$ between the respective encoders and the decoder in the original problem. 
Encoder $j \in \{1,2\}$ observes $(k_j,f_j,x_j^n)$, and generates $(u_{j}^n,v_{j}^n)$ according to the p.m.f. $P(u_{j}^n,v_{j}^n|k_j,f_j,x_j^n)$ from the binning protocol above. Further, encoder $j \in \{1,2\}$ then creates the channel input $\tilde{x}_j^n$ according to the p.m.f. $p(\tilde{x}_j^n|v_{j}^n)$ for transmission over the channel, while the channel outputs $(\tilde{y}^n,\tilde z^n)$.
The legitimate decoder first estimates $(\hat{u}_{1}^n,\hat{u}_{2}^n)$ from its observations $(k_1,k_2,f_1,f_2,w^n,\tilde{y}^n)$ using the Slepian-Wolf decoder from the binning protocol. 
Then it generates the output $y^n$ according to the distribution $p_{Y^n|U_{1}^n,U_{2}^n,W^n}(y^n|\hat{u}_{1}^n,\hat{u}_{2}^n,w^n)$. 
The induced random p.m.f. from the random coding scheme is
\begin{align}
&\hat{P}(x_1^n,x_2^n,w^n,y^n,u_{1}^n,u_{2}^n,v_{1}^n,v_{2}^n,\tilde{x}_1^n,\tilde{x}_2^n,\tilde{y}^n,\tilde z^n,k_1,f_1,k_2,f_2,\hat{u}_{1}^n,\hat{u}_{2}^n) \notag\\
&= p^{\text{U}}(k_1)p^{\text{U}}(f_1)p^{\text{U}}(k_2)p^{\text{U}}(f_2)p(x_1^n,x_2^n,w^n)  P(u_{1}^n,v_{1}^n|k_1,f_1,x_1^n) P(u_{2}^n,v_{2}^n|k_2,f_2,x_2^n) \notag\\
&\phantom{w} \times p(\tilde{x}_1^n|v_{1}^n)p(\tilde{x}_2^n|v_{2}^n)p(\tilde{y}^n,\tilde z^n|\tilde{x}_1^n,\tilde{x}_2^n) P^{SW}(\hat{u}_{1}^n,\hat{u}_{2}^n|k_1,f_1,k_2,f_2,w^n,\tilde{y}^n) p(y^n|\hat{u}_{1}^n,\hat{u}_{2}^n,w^n). \label{eq:rcpmfnoisy}
\end{align}

We next establish constraints that lead to nearly identical induced p.m.f.'s from both protocols. 

\noindent \underline{Analysis of Rate Constraints:}\\
Using the fact that $(k_j,f_j)$ are bin indices of $(u_{j}^n,v_{j}^n)$ for $j \in \{1,2\}$, we impose the conditions
\begin{align}
R_{01}+\tilde{R}_1 &\leq H(U_{1},V_{1}|X_1,X_2,W) \notag\\
&= H(U_{1}|X_1)+H(V_{1}), \label{eq:cond11noisy} \\
R_{02}+\tilde{R}_2 &\leq H(U_{2},V_{2}|X_1,X_2,W) \notag\\
&= H(U_{2}|X_2)+H(V_{2}), \label{eq:cond12noisy} \\
R_{01}+\tilde{R}_1+R_{02}+\tilde{R}_2 &\leq H(U_{1},V_{1},U_{2},V_{2}|X_1,X_2,W)\nonumber\\
&=H(U_{1},U_{2}|X_1,X_2)\!+\!H(V_{1},V_{2}) \label{eq:cond13noisy}
\end{align}
{(where the equalities in \eqref{eq:cond11noisy}--\eqref{eq:cond13noisy} follow from the independence between $(U_{1},U_{2},X_1,X_2,W,Y)$ and $(V_{1},V_{2},\tilde{X}_1,\tilde{X}_2,\tilde{Y},\tilde Z)$ and the Markov chains $U_{1} \fooA X_1 \fooA (U_{2},X_2,W)$, $U_{2} \fooA X_2 \fooA (U_{1},X_1,W)$, and $(U_{1},U_{2}) \fooA (X_1,X_2) \fooA W$)}, that ensure, by invoking \cite[Theorem 1]{yassaee2014achievability}
\begin{align}
P(x_1^n,x_2^n,w^n,k_1,k_2,f_1,f_2)  &\approx p^{\text{U}}(k_1)p^{\text{U}}(f_1)p^{\text{U}}(k_2)p^{\text{U}}(f_2)p(x_1^n,x_2^n,w^n) \notag\\
&= \hat{P}(x_1^n,x_2^n,w^n,k_1,k_2,f_1,f_2).
\end{align}
This also implies
\begin{align}
&P(x_1^n,x_2^n,w^n,u_{1}^n,u_{2}^n,v_{1}^n,v_{2}^n,\tilde{x}_1^n,\tilde{x}_2^n,\tilde{y}^n,\tilde z^n,k_1,f_1,k_2,f_2,\hat{u}_{1}^n,\hat{u}_{2}^n) \notag\\
&\phantom{w} \approx \hat{P}(x_1^n,x_2^n,w^n,u_{1}^n,u_{2}^n,v_{1}^n,v_{2}^n,\tilde{x}_1^n,\tilde{x}_2^n,\tilde{y}^n,\tilde z^n,k_1,f_1,k_2,f_2,\hat{u}_{1}^n,\hat{u}_{2}^n). \label{eq:condit1noisy}
\end{align}

We next impose the following constraints for the success of the (Slepian-Wolf) decoder by Slepian-Wolf theorem~\cite{slepian1973noiseless}
\begin{align}
R_{01}+\tilde{R}_1 &\geq H(U_{1},V_{1}|U_{2},V_{2},W,\tilde{Y}) \notag\\
&= H(U_{1}|U_{2},W)+H(V_{1}|V_{2},\tilde{Y}), \label{eq:cond21noisy} \\
R_{02}+\tilde{R}_2 &\geq H(U_{2},V_{2}|U_{1},V_{1},W,\tilde{Y}) \notag\\
&= H(U_{2}|U_{1},W)+H(V_{2}|V_{1},\tilde{Y}), \label{eq:cond22noisy} \\
R_{01}+\tilde{R}_1+R_{02}+\tilde{R}_2 &\geq H(U_{1},V_{1},U_{2},V_{2}|W,\tilde{Y}) \notag\\
&= H(U_{1},U_{2}|W)+H(V_{1},V_{2}|\tilde{Y}). \label{eq:cond23noisy}
\end{align}
Expressions \eqref{eq:cond21noisy}--\eqref{eq:cond23noisy} suffice to obtain
\begin{align}
&P(x_1^n,x_2^n,w^n,u_{1}^n,u_{2}^n,v_{1}^n,v_{2}^n,\tilde{x}_1^n,\tilde{x}_2^n,\tilde{y}^n,\tilde z^n,k_1,f_1,k_2,f_2,\hat{u}_{1}^n,\hat{u}_{2}^n) \notag\\
&\phantom{w} \approx P(x_1^n,x_2^n,w^n,u_{1}^n,u_{2}^n,v_{1}^n,v_{2}^n,\tilde{x}_1^n,\tilde{x}_2^n,\tilde{y}^n,\tilde z^n,k_1,f_1,k_2,f_2) \mathbbm{1}\{\hat{u}_{1}^n=u_{1}^n,\hat{u}_{2}^n=u_{2}^n\}. \label{eq:condit2noisy}
\end{align}
Using \eqref{eq:condit2noisy} and \eqref{eq:condit1noisy} in conjunction with the first and third parts of \cite[Lemma 4]{yassaee2014achievability}, we obtain
\begin{align}
&\hat{P}(x_1^n,x_2^n,w^n,u_{1}^n,u_{2}^n,v_{1}^n,v_{2}^n,\tilde{x}_1^n,\tilde{x}_2^n,\tilde{y}^n,\tilde z^n,k_1,f_1,k_2,f_2,\hat{u}_{1}^n,\hat{u}_{2}^n,y^n) \notag\\
&= \hat{P}(x_1^n,x_2^n,w^n,u_{1}^n,u_{2}^n,v_{1}^n,v_{2}^n,\tilde{x}_1^n,\tilde{x}_2^n,\tilde{y}^n,\tilde z^n,k_1,f_1,k_2,f_2,\hat{u}_{1}^n,\hat{u}_{2}^n) p(y^n|\hat{u}_{1}^n,\hat{u}_{2}^n,w^n) \notag\\
&\approx P(x_1^n,x_2^n,w^n,u_{1}^n,u_{2}^n,v_{1}^n,v_{2}^n,\tilde{x}_1^n,\tilde{x}_2^n,\tilde{y}^n,\tilde z^n,k_1,f_1,k_2,f_2)p(y^n|\hat{u}_{1}^n,\hat{u}_{2}^n,w^n) \notag\\
&\phantom{w} \times \mathbbm{1}\{\hat{u}_{1}^n=u_{1}^n,\hat{u}_{2}^n=u_{2}^n\} \notag\\
&= P(x_1^n,x_2^n,w^n,u_{1}^n,u_{2}^n,v_{1}^n,v_{2}^n,\tilde{x}_1^n,\tilde{x}_2^n,\tilde{y}^n,\tilde z^n,k_1,f_1,k_2,f_2)p(y^n|u_{1}^n,u_{2}^n,w^n) \notag\\
&\phantom{w} \times \mathbbm{1}\{\hat{u}_{1}^n=u_{1}^n,\hat{u}_{2}^n=u_{2}^n\}  \notag\\
&= P(x_1^n,x_2^n,w^n,u_{1}^n,u_{2}^n,v_{1}^n,v_{2}^n,\tilde{x}_1^n,\tilde{x}_2^n,\tilde{y}^n,\tilde z^n,k_1,f_1,k_2,f_2,y^n) \mathbbm{1}\{\hat{u}_{1}^n=u_{1}^n,\hat{u}_{2}^n=u_{2}^n\}.
\end{align}
This implies, by the first part of \cite[Lemma 4]{yassaee2014achievability}
\begin{align}
&\hat{P}(x_1^n,x_2^n,w^n,y^n,\tilde z^n,f_1,f_2) \approx P(x_1^n,x_2^n,w^n,y^n,\tilde z^n,f_1,f_2). \label{eq:osrbcond1noisy}
\end{align}

To show the strong secrecy condition, we must also ensure that $(X_1^n,X_2^n,W^n,Y^n,\tilde Z^n)$ is nearly independent of the additional shared randomness indices $(F_1,F_2)$ assumed in the protocol, so that the latter can be eliminated without affecting the desired i.i.d. distribution. This ensures that for almost all choices $(F_1=f_1,F_2=f_2)$, the term $I(X_1^n,X_2^n,W^n,Y^n;\tilde Z^n|F_1=f_1,F_2=f_2)$ is asymptotically zero. The nodes can then agree upon an instance $(F_1=f_1^{*},F_2=f_2^{*})$, using randomness extraction results such as the following from~\cite[Section III.A]{pierrot2013joint},~\cite[Lemma 2.20]{cervia2018coordination}.
\begin{lemma} \cite{pierrot2013joint,cervia2018coordination} \label{lem:exp}
Let $A^n$ be a discrete memoryless source drawn according to $P_{A^n}$, and $P_{B^n|A^n}$ be a discrete memoryless channel. Consider a uniform random binning of $B^n$ given by $K=\phi_n(B^n)$, where $\phi_n: \mathcal{B}^n \to [1:2^{nR}]$. Then if $R \leq H(B|A)$, there exists a constant $\alpha > 0$ such that
\begin{align}
\mathbb{E}_{\phi_n}[D(P_{A^n, K} \parallel P_{A^n} \cdot p^{U}(K))] \leq 2^{-\alpha n},
\end{align}
where $D(\cdot\parallel\cdot)$ stands for the Kullback-Leibler (KL) divergence between the two distributions. Notice that a vanishing KL divergence also ensures a vanishing total variation distance using Pinsker's inequality. 
\end{lemma}
Using Lemma~\ref{lem:exp}, we impose the following conditions:
\begin{align}
\tilde{R}_1 &\leq H(U_{1},V_{1}|X_1,X_2,W,Y,\tilde Z) \notag\\
&= H(U_{1}|X_1,X_2,W,Y)+H(V_{1}|\tilde Z), \label{eq:cond31secnoisy} \\
\tilde{R}_2 &\leq H(U_{2},V_{2}|X_1,X_2,W,Y,\tilde Z) \notag\\
&= H(U_{2}|X_1,X_2,W,Y)+H(V_{2}|\tilde Z), \label{eq:cond32secnoisy} \\
\tilde{R}_1+\tilde{R}_2 &\leq H(U_{1},V_{1},U_{2},V_{2}|X_1,X_2,W,Y,\tilde Z)  \notag\\
&= H(U_{1},U_{2}|X_1,X_2,W,Y)+H(V_{1},V_{2}|\tilde Z). \label{eq:cond33secnoisy}
\end{align}
This suffices to obtain
\begin{align}
&P(x_1^n,x_2^n,w^n,y^n,\tilde z^n,f_1,f_2) \approx p^{\text{U}}(f_1) p^{\text{U}}(f_2) p(x_1^n,x_2^n,w^n,y^n,\tilde z^n), 
\end{align}
which implies that
\begin{align}
&\hat{P}(x_1^n,x_2^n,w^n,y^n,\tilde z^n,f_1,f_2) \approx p^{\text{U}}(f_1) p^{\text{U}}(f_2) p(x_1^n,x_2^n,w^n,y^n,\tilde z^n),
\end{align}
by invoking \eqref{eq:osrbcond1noisy} and the triangle inequality.
Thus, equations \eqref{eq:cond11noisy} -- \eqref{eq:cond13noisy}, \eqref{eq:cond21noisy} -- \eqref{eq:cond23noisy} and \eqref{eq:cond31secnoisy} -- \eqref{eq:cond33secnoisy} guarantee that
\begin{align}
P(x_1^n,x_2^n,w^n,y^n,\tilde z^n,f_1,f_2) &\approx p^{\text{U}}(f_1) p^{\text{U}}(f_2) p(x_1^n,x_2^n,w^n,y^n,\tilde z^n), \label{osrbref1noisy} \\
P(x_1^n,x_2^n,w^n,y^n,\tilde z^n,f_1,f_2)  & \approx \hat{P}(x_1^n,x_2^n,w^n,y^n,\tilde z^n,f_1,f_2) \label{osrbref2noisy}.
\end{align}
Conditions \eqref{osrbref1noisy} and \eqref{osrbref2noisy} imply
\begin{align}
&\lim_{n \to \infty} \mathbb{E} \biggl[\lVert P(x_1^n,x_2^n,w^n,y^n,\tilde z^n,f_1,f_2)-p^{\text{U}}(f_1) p^{\text{U}}(f_2) p(x_1^n,x_2^n,w^n,y^n,\tilde z^n)\rVert_1  \notag\\ 
&\phantom{wwww} + \lVert P(x_1^n,x_2^n,w^n,y^n,\tilde z^n,f_1,f_2)-\hat{P}(x_1^n,x_2^n,w^n,y^n,\tilde z^n,f_1,f_2) \rVert_1\biggr] = 0,
\end{align}
where the expectation is over the random binning. This in turn implies the existence of a realization of the random binning with corresponding p.m.f. $p$ so that we can replace $P$ with $p$ and denote the resulting p.m.f.'s for the random binning and random coding protocols by $p$ and $\hat{p}$, then
\begin{align}
p(x_1^n,x_2^n,w^n,y^n,\tilde z^n,f_1,f_2) &\stackrel{\epsilon_n}\approx p^{\text{U}}(f_1) p^{\text{U}}(f_2) p(x_1^n,x_2^n,w^n,y^n,\tilde z^n), \label{osrbref3noisy} \\
p(x_1^n,x_2^n,w^n,y^n,\tilde z^n,f_1,f_2)  & \stackrel{\delta_n}\approx \hat{p}(x_1^n,x_2^n,w^n,y^n,\tilde z^n,f_1,f_2) \label{osrbref4noisy},
\end{align}
where $\epsilon_n \to 0$ and $\delta_n \to 0$ as $n \to \infty$.
Condition \eqref{osrbref3noisy} implies that $p(f_1,f_2) \stackrel{\epsilon_n}\approx p^{\text{U}}(f_1) p^{\text{U}}(f_2)$. Also, we have $\hat{p}(f_1,f_2) = p^{\text{U}}(f_1) p^{\text{U}}(f_2)$. Thus, \eqref{osrbref3noisy} and \eqref{osrbref4noisy} imply
\begin{align}
p(f_1,f_2) p(x_1^n,x_2^n,w^n,y^n,\tilde z^n|f_1,f_2) &\stackrel{2\epsilon_n}\approx p(f_1,f_2) p(x_1^n,x_2^n,w^n,y^n,\tilde z^n), \\
p(f_1,f_2)p(x_1^n,x_2^n,w^n,y^n,\tilde z^n|f_1,f_2) &\stackrel{\epsilon_n+\delta_n}\approx p(f_1,f_2)\hat{p}(x_1^n,x_2^n,w^n,y^n,\tilde z^n|f_1,f_2).
\end{align}
In terms of total variation distances, this in turn gives
\begin{align}
&\sum_{f_1,f_2} p(f_1,f_2)||p(x_1^n,x_2^n,w^n,y^n,\tilde z^n|f_1,f_2)-p(x_1^n,x_2^n,w^n,y^n,\tilde z^n)||_1 \leq 2\epsilon_n, \label{osrbref5noisy}\\
&\sum_{f_1,f_2} p(f_1,f_2)||p(x_1^n,x_2^n,w^n,y^n,\tilde z^n|f_1,f_2)-\hat{p}(x_1^n,x_2^n,w^n,y^n,\tilde z^n|f_1,f_2)||_1 \leq \epsilon_n+\delta_n. \label{osrbref6noisy}
\end{align}
Adding \eqref{osrbref5noisy} and \eqref{osrbref6noisy}, it follows that there exist instances $f_1^{*},f_2^{*}$ such that $p(f_1^{*}), p(f_2^{*}) > 0$ and
\begin{align}
p(x_1^n,x_2^n,w^n,y^n,\tilde z^n|f_1^{*},f_2^{*}) & \stackrel{3\epsilon_n+\delta_n}\approx p(x_1^n,x_2^n,w^n,y^n,\tilde z^n), \label{osrbref7noisy} \\
p(x_1^n,x_2^n,w^n,y^n,\tilde z^n|f_1^{*},f_2^{*}) &\stackrel{3\epsilon_n\!+\!\delta_n}\approx \hat{p}(x_1^n,x_2^n,w^n,y^n,\tilde z^n|f_1^{*},f_2^{*}) \label{osrbref8noisy}.
\end{align}
From \eqref{osrbref8noisy} and \eqref{osrbref7noisy}, we obtain
\begin{align}
&\hat{p}(x_1^n,x_2^n,w^n,y^n,\tilde z^n|f_1^{*},f_2^{*}) \approx p(x_1^n,x_2^n,w^n,y^n,\tilde z^n).
\end{align}
This proves the strong coordination of the joint p.m.f. of $(X_1^n,X_2^n,W^n,Y^n)$. For the secrecy constraint, we note that since conditions \eqref{eq:cond31secnoisy}--\eqref{eq:cond33secnoisy} ensure the asymptotic mutual independence of $(X_1^n,X_2^n,W^n,Y^n)$, $\tilde Z^n$ and $(F_1,F_2)$, we have the following sequence of steps as $n \to \infty$
\begin{align}
&I_p(F_1,F_2,\tilde Z^n;X_1^n,X_2^n,W^n,Y^n) \to 0 \notag\\
&\implies I_p(\tilde Z^n;X_1^n,X_2^n,W^n,Y^n|F_1,F_2) \to 0 \notag\\
&\implies I_p(\tilde Z^n;X_1^n,X_2^n,W^n,Y^n|F_1=f_1,F_2=f_2) \to 0 \:\: \forall \:\: f_1,f_2 \:\: \textrm{s.t.}  \:\: p(f_1),p(f_2)>0 \notag\\
&\implies I(\tilde Z^n;X_1^n,X_2^n,W^n,Y^n)|p(x_1^n,x_2^n,w^n,y^n,\tilde z^n|f_1,f_2)  \to 0 \:\: \forall \:\: f_1,f_2 \:\: \textrm{s.t.}  \:\: p(f_1),p(f_2)>0 \notag\\
&\implies I(\tilde Z^n;X_1^n,X_2^n,W^n,Y^n)|p(x_1^n,x_2^n,w^n,y^n,\tilde z^n|f_1^{*},f_2^{*})  \to 0, \label{eq:sec2noisy}
\end{align}
where $f_1^{*},f_2^{*}$ is fixed in \eqref{osrbref7noisy}. 
Thus, we have
\begin{align}
&I(\tilde Z^n;X_1^n,X_2^n,W^n,Y^n)|p(x_1^n,x_2^n,w^n,y^n,\tilde z^n|f_1^{*},f_2^{*}) \to 0 \:\: \textrm{as} \:\: n \to \infty, \label{refeq1noisy}\\
&p(x_1^n,x_2^n,w^n,y^n,\tilde z^n|f_1^{*},f_2^{*}) \approx \hat{p}(x_1^n,x_2^n,w^n,y^n,\tilde z^n|f_1^{*},f_2^{*}). \label{refeq3noisy}
\end{align}
Since mutual information is a continuous function of the probability distribution, expressions \eqref{refeq1noisy}--\eqref{refeq3noisy} imply 
\begin{align}
&I(\tilde Z^n;X_1^n,X_2^n,W^n,Y^n)|\hat{p}(x_1^n,x_2^n,w^n,y^n,\tilde z^n|f_1^{*},f_2^{*})   \to 0
\end{align}
as $n \to \infty$, by utilizing the property that if random variables $\Theta$ and $\Theta'$ with the same support set $\varTheta$ satisfy $\lVert p_{\Theta}-p_{\Theta'}\rVert_1 \leq \epsilon \leq 1/4$, then according to \cite[Theorem 17.3.3]{cover2012elements}, it follows that $|H(\Theta)-H(\Theta')| \leq \zeta\log|\varTheta|$, where $\zeta$ approaches zero as $\epsilon$ approaches zero.
Finally on eliminating $(\tilde{R}_1,\tilde{R}_2)$ from equations \eqref{eq:cond11noisy} -- \eqref{eq:cond13noisy}, \eqref{eq:cond21noisy} -- \eqref{eq:cond23noisy} and \eqref{eq:cond31secnoisy} -- \eqref{eq:cond33secnoisy} by the Fourier-Motzkin elimination (FME) procedure, we obtain the rate constraints in Theorem~\ref{thm:encsideIBnoisy}.

\section{Proof of Theorem \ref{thm:encsideOBnoisy}}\label{proof:OB-thm6}
Consider a coding scheme that induces a joint distribution on $(X_1^n,X_2^n,W^n,Y^n,\tilde Z^n)$ which satisfies the constraint
\begin{gather}
\lVert p_{X_1^n,X_2^n,W^n,Y^n,\tilde Z^n}-p_{\tilde Z^n} \cdot q^{(n)}_{X_1X_2WY}\rVert_1 \leq \epsilon, \label{eqn:totalSRnoisy}
\end{gather}
for $\epsilon\in(0,\frac{1}{4}]$. 
To simplify notation, we define $\Theta_{\sim i} \triangleq (\Theta^{i-1}, \Theta_{i+1}^n)$ for any vector $\Theta^n$. The following lemmas will be useful in establishing the outer bound.
\begin{lemma} \cite[Lemma 6]{CerviaLLB20} \label{lem:c1}
Let $p_{S^n}$ be such that $||p_{S^n}-q_S^{(n)}||_1 \leq \epsilon$, where $q^{(n)}_S(s^n)=\prod_{i=1}^n q_S(s_i)$, then
\begin{align}
\sum_{i=1}^n I_p(S_i;S_{\sim i}) \leq ng_1(\epsilon),
\end{align}
where ${ g_1(\epsilon)=2\sqrt{\epsilon}\left(H(S)+\log{|\mathcal{S}|+\log{\frac{1}{\sqrt{\epsilon}}}}\right)}\to 0$ as $\epsilon \to 0$.
\end{lemma}

\begin{lemma} \cite[Lemma VI.3]{cuff2013distributed} \label{lem:c2}
Let $p_{S^n}$ be such that $||p_{S^n}-q_S^{(n)}||_1 \leq \epsilon$, where $q^{(n)}_S(s^n)=\prod_{i=1}^n q_S(s_i)$, then for any RV $T \in [1:n]$ independent of $S^n$,
\begin{align}
I_p(T;S_T) \leq g_2(\epsilon),
\end{align}
where ${ g_2(\epsilon)=4\epsilon\left(\log{|\mathcal{S}|}+\log{\frac{1}{\epsilon}}\right)} \to 0$ as $\epsilon \to 0$.
\end{lemma}
We now take $S=(X_1,X_2,W,Y)$ in Lemmas~\ref{lem:c1} and~\ref{lem:c2}.
{ Notice that for $\epsilon\in(0,\frac{1}{4}]$, one obtains
\begin{align}
\max\{g_1(\epsilon),g_2(\epsilon)\}&\leq g(\epsilon) \notag\\
&:=2\sqrt{\epsilon}\left(H(S)+\log|\mathcal{\tilde Z}|+\log{|\mathcal{S}|+2\log{\frac{1}{\sqrt{\epsilon}}}}\right) \notag\\
&=2\sqrt{\epsilon}\biggl(H(X_1,X_2,W,Y)+\log{|\mathcal{X}_1||\mathcal{X}_2||\mathcal{W}||\mathcal{Y}||\mathcal{\tilde Z}|+2\log{\frac{1}{\sqrt{\epsilon}}}}\biggr),
\end{align}
which is exactly the definition of $g(\epsilon)$ in Theorem~\ref{thm:encsideOBnoisy}.
Thus, one can replace $g_1(\epsilon)$ and $g_2(\epsilon)$ in Lemmas~\ref{lem:c1} and \ref{lem:c2} by $g(\epsilon)$, which satisfies $\lim_{\epsilon\to 0} g(\epsilon) = 0$ as well.}


We first prove the strong secrecy constraint. We can bound the total variation distance between the induced joint distribution on the random variables $(X_1^n,X_2^n,W^n,Y^n,\tilde Z^n)$ and the product of the induced distributions on $(X_1^n,X_2^n,W^n,Y^n)$ and $\tilde Z^n$ as follows:
\begin{align}
\lVert p_{X_1^n,X_2^n,W^n,Y^n,\tilde Z^n}-p_{X_1^n,X_2^n,W^n,Y^n}p_{\tilde Z^n}\rVert_1 &\stackrel{(a)}\leq \lVert p_{X_1^n,X_2^n,W^n,Y^n,\tilde Z^n}-p_{\tilde Z^n}q^{(n)}_{X_1X_2WY}\rVert_1 \notag\\
&\phantom{www}+\lVert p_{X_1^n,X_2^n,W^n,Y^n}p_{\tilde Z^n}-p_{\tilde Z^n}q^{(n)}_{X_1X_2WY}\rVert_1 \notag\\
&= \lVert p_{X_1^n,X_2^n,W^n,Y^n,\tilde Z^n}-p_{\tilde Z^n}q^{(n)}_{X_1X_2WY}\rVert_1 \notag\\
&\phantom{www}+\lVert p_{X_1^n,X_2^n,W^n,Y^n}-q^{(n)}_{X_1X_2WY}\rVert_1 \notag\\
&\stackrel{(b)}\leq 2\lVert p_{X_1^n,X_2^n,W^n,Y^n,\tilde Z^n}-p_{\tilde Z^n}q^{(n)}_{X_1X_2WY}\rVert_1 \notag\\
&\stackrel{(c)}\leq 2\epsilon, \label{eq:refe00inoisy}
\end{align}
where (a) follows from triangle inequality, (b) follows from \cite[Lemma V.1]{cuff2013distributed}, while (c) follows by the definition of a successful code, i.e., from \eqref{eqn:totalSRnoisy}. Consequently, it follows from \cite[Theorem 17.3.3]{cover2012elements} that the mutual information $I(X_1^n,X_2^n,W^n,Y^n;\tilde Z^n)$ of interest can be bounded as
\begin{align}
&I(X_1^n,X_2^n,W^n,Y^n;\tilde Z^n) \notag\\
&\phantom{ww}=H(X_1^n,X_2^n,W^n,Y^n)+H(\tilde Z^n)-H(X_1^n,X_2^n,W^n,Y^n,\tilde Z^n) \notag\\
&\phantom{ww}\leq 4n\epsilon\biggl(\log|\mathcal{X}_1|+\log|\mathcal{X}_2|+\log|\mathcal{W}|+\log|\mathcal{Y}| +\log|\mathcal{\tilde Z}|+\log\frac{1}{\epsilon}\biggr) \notag\\
&\phantom{ww}\leq ng(\epsilon). \label{eq:refe0inoisy}
\end{align} 
We now prove the first inequality in Theorem~\ref{thm:encsideOBnoisy}. Note that by the functional representation lemma (FRL)~\cite{el2011network}, the randomized encoders $p^{\emph{Enc}_j}(\tilde{x}_j^n|x_j^n,k_j)$ for $j \in \{1,2\}$ in Definition~\ref{defn:codenoisy} can be expressed as deterministic mappings such that
\begin{align}
H(\tilde{X}_j^n|K_j,X_j^n,\Theta_j)=0,
\end{align}
where $(\Theta_1,\Theta_2)$ is independent of $(K_1,K_2,X_1^n,X_2^n)$.
Then we have following the chain of inequalities:
\begin{align}
I(X_1^n;\tilde{Y}^n|K_1,K_2,X_2^n,W^n,\Theta_1,\Theta_2) &= \sum_{i=1}^n I(X_{1i};\tilde{Y}^n|X_{1,i+1}^n,K_1,K_2,X_2^n,W^n,\Theta_1,\Theta_2) \notag\\
&\stackrel{(a)}= \sum_{i=1}^n I(X_{1i};\tilde{Y}^n,X_{1,i+1}^n,K_1,X_{2}^{i-1},W_{\sim i}, \notag\\
&\phantom{wwwwwwwwwwwwwwww}\Theta_1,\Theta_2|X_{2i},W_i,K_2,X_{2,i+1}^n) \notag\\
&\geq \sum_{i=1}^n I(X_{1i};\tilde{Y}^n,W_{\sim i},X_{1,i+1}^n,K_1|X_{2i},W_i,K_2,X_{2,i+1}^n) \notag\\
&\stackrel{(b)}= \sum_{i=1}^n I(X_{1i};U_{1i}|U_{2i},X_{2i},W_i) \notag\\
&= nI(X_{1T};U_{1T}|U_{2T},X_{2T},W_T,T) \notag\\
&= nI(U_{1};X_1|U_{2},X_{2},W,T), \label{refnoisy3}
\end{align}
where (a) follows due to the independence of $(K_{[1:2]},\Theta_{[1:2]})$ from $(X_1^n,X_2^n,W^n)$ and the joint i.i.d. nature of $(X_{1i},X_{2i},W_i),i=1,\dots,n,$ while (b) follows from an auxiliary random variable identification $U_{1i}=(K_1,X_{1,i+1}^n,\tilde{Y}^n,W_{\sim i})$ and $U_{2i}=(K_2,X_{2,i+1}^n)$.

We next bound $I(X_1^n;\tilde{Y}^n|K_1,K_2,X_2^n,W^n,\Theta_1,\Theta_2)$ in another way as follows:
\begin{align}
I(X_1^n;\tilde{Y}^n|K_1,K_2,X_2^n,W^n,\Theta_1,\Theta_2) &= H(\tilde{Y}^n|K_1,K_2,X_2^n,W^n,\Theta_1,\Theta_2) \notag\\
&\phantom{www}-H(\tilde{Y}^n|K_1,K_2,X_1^n,X_2^n,W^n,\Theta_1,\Theta_2) \notag\\
&\leq H(\tilde{Y}^n|K_2,X_2^n,\Theta_2) \notag\\
&\phantom{www}-H(\tilde{Y}^n|K_1,K_2,X_1^n,X_2^n,W^n,\tilde{X}_1^n,\tilde{X}_2^n,\Theta_1,\Theta_2) \notag\\
&\stackrel{(a)}= H(\tilde{Y}^n|K_2,X_2^n,\Theta_2,\tilde{X}_2^n)-H(\tilde{Y}^n|\tilde{X}_1^n,\tilde{X}_2^n) \notag\\
&\leq H(\tilde{Y}^n|\tilde{X}_2^n)-H(\tilde{Y}^n|\tilde{X}_1^n,\tilde{X}_2^n) \notag\\
&\stackrel{(b)}\leq \sum_{i=1}^n H(\tilde{Y}_i|\tilde{X}_{2i})-\sum_{i=1}^n H(\tilde{Y}_i|\tilde{X}_{1i},\tilde{X}_{2i}) \notag\\
&= nI(\tilde{X}_{1T};\tilde{Y}_T|\tilde{X}_{2T},T) = nI(\tilde{X}_{1};\tilde{Y}|\tilde{X}_{2},T), \label{refnoisy4}
\end{align}
where (a) follows since $\tilde{X}_2^n$ is a deterministic function of $(K_2,X_2^n,\Theta_2)$ and since $\tilde{Y}^n \fooA (\tilde{X}_1^n,\tilde{X}_2^n) \fooA (K_1,K_2,X_1^n,X_2^n,W^n,\Theta_1,\Theta_2)$ is a Markov chain. Here (b) follows since conditioning does not increase the entropy and the memoryless nature of the channel $p(\tilde{y}|\tilde{x}_1,\tilde{x}_2)$. 
From \eqref{refnoisy3} and \eqref{refnoisy4}, it follows that
\begin{align}
I(\tilde{X}_{1};\tilde{Y}|\tilde{X}_{2},T) \geq I(U_{1};X_1|U_{2},X_2,W,T).
\end{align}
The inequality $I(\tilde{X}_{2};\tilde{Y}|\tilde{X}_{1},T) \geq I(U_{2};X_2|U_{1},X_1,W,T)$ follows analogously. 
We prove the third inequality as follows. Since $\tilde{Y}^n \fooA (\tilde{X}_1^n,\tilde{X}_2^n) \fooA (K_1,K_2,X_1^n,X_2^n,W^n)$ is a Markov chain, we can write by data processing inequality~\cite{cover2012elements}
\begin{align}
0&\leq I(\tilde{X}_1^n,\tilde{X}_2^n;\tilde{Y}^n)-I(K_1,K_2,X_1^n,X_2^n,W^n;\tilde{Y}^n) \notag\\
&\leq I(\tilde{X}_1^n,\tilde{X}_2^n;\tilde{Y}^n)-I(X_1^n,X_2^n;\tilde{Y}^n|K_1,K_2,W^n) \notag\\
&\stackrel{(a)}\leq \sum_{i=1}^n I(\tilde{X}_{1i},\tilde{X}_{2i};\tilde{Y}_i)-\sum_{i=1}^n I(X_{1i},X_{2i};\tilde{Y}^n|K_1,K_2,X_{1,i+1}^n,X_{2,i+1}^n,W^n) \notag\\
&\stackrel{(b)}= \sum_{i=1}^n I(\tilde{X}_{1i},\tilde{X}_{2i};\tilde{Y}_i)-\sum_{i=1}^n I(X_{1i},X_{2i};\tilde{Y}^n,K_1,K_2,X_{1,i+1}^n,X_{2,i+1}^n,W_{\sim i}|W_i) \notag\\
&= \sum_{i=1}^n I(\tilde{X}_{1i},\tilde{X}_{2i};\tilde{Y}_i)-\sum_{i=1}^n I(X_{1i},X_{2i};U_{1i},U_{2i}|W_i) \notag\\
&= nI(\tilde{X}_{1T},\tilde{X}_{2T};\tilde{Y}_T|T)-nI(X_{1T},X_{2T};U_{1T},U_{2T}|W_T,T) \notag\\
&= nI(\tilde{X}_{1},\tilde{X}_{2};\tilde{Y}|T)-nI(U_{1},U_{2};X_{1},X_{2}|W,T),
\end{align}
where (a) follows from the memoryless nature of the channel $p(\tilde{y}|\tilde{x}_1,\tilde{x}_2)$, while (b) follows due to the independence of $(K_1,K_2)$ from $(X_1^n,X_2^n,W^n)$ and the joint i.i.d. nature of $(X_{1i},X_{2i},W_i),i=1,\dots,n$.

We next prove the sum-rate constraint on $(R_{01}+R_{02})$.
\begin{align}
n(R_{01}+R_{02}) &=H(K_1,K_2) \notag\\
&\geq H(K_1,K_2|\tilde{Y}^n,W^n) \notag\\
&\geq I(K_1,K_2;X_1^n,X_2^n,Y^n|\tilde{Y}^n,W^n) \notag\\
&= I(K_1,K_2,\tilde{Y}^n;X_1^n,X_2^n,Y^n|W^n)-I(\tilde{Y}^n;X_1^n,X_2^n,Y^n|W^n) \notag\\
&\geq I(K_1,K_2,\tilde{Y}^n;X_1^n,X_2^n,Y^n|W^n)-I(\tilde{Y}^n;X_1^n,X_2^n,W^n,Y^n) \notag\\
&\phantom{ww}+I(\tilde Z^n;X_1^n,X_2^n,W^n,Y^n)-I(\tilde Z^n;X_1^n,X_2^n,W^n,Y^n)\notag\\
&\stackrel{(a)}\geq I(K_1,K_2,\tilde{Y}^n;X_1^n,X_2^n,Y^n|W^n)-ng(\epsilon) \notag\\
&\phantom{ww}-I(\tilde{Y}^n;X_1^n,X_2^n,W^n,Y^n)+I(\tilde Z^n;X_1^n,X_2^n,W^n,Y^n) \notag\\
&= I(K_1,K_2,\tilde{Y}^n;X_1^n,X_2^n,Y^n|W^n)-ng(\epsilon) \notag\\
&\phantom{ww}-I(\tilde{Y}^n;X_1^n,X_2^n,W^n)+I(\tilde Z^n;X_1^n,X_2^n,W^n) \notag\\
&\phantom{ww}-I(\tilde{Y}^n;Y^n|X_1^n,X_2^n,W^n)+I(\tilde Z^n;Y^n|X_1^n,X_2^n,W^n) \notag\\
&\geq I(K_1,K_2,\tilde{Y}^n;X_1^n,X_2^n,Y^n|W^n)-ng(\epsilon) \notag\\
&\phantom{w}-I(\tilde{Y}^n;X_1^n,X_2^n,W^n|\tilde Z^n)\!-\!H(Y^n|X_1^n,X_2^n,W^n) \notag\\
&\stackrel{(b)}\geq I(K_1,K_2,\tilde{Y}^n;X_1^n,X_2^n,Y^n|W^n)-ng(\epsilon) \notag\\
&\phantom{ww}-I(\tilde{Y}^n;\tilde{X}_1^n,\tilde{X}_2^n|\tilde Z^n)-H(Y^n|X_1^n,X_2^n,W^n) \notag\\
&\geq \sum_{i=1}^n I(K_1,K_2,\tilde{Y}^n;X_{1i},X_{2i},Y_i|X_{1,i+1}^n,X_{2,i+1}^n,Y_{i+1}^n,W^n) \notag\\
&\phantom{ww}-\sum_{i=1}^n I(\tilde{Y}_i;\tilde{X}_1^n,\tilde{X}_2^n|\tilde Z^n,\tilde{Y}^{i-1}) -\sum_{i=1}^n H(Y_i|X_{1i},X_{2i},W_i) -ng(\epsilon) \notag\\
&\stackrel{(c)}= \sum_{i=1}^n I(K_1,K_2,\tilde{Y}^n,X_{1,i+1}^n,X_{2,i+1}^n,Y_{i+1}^n,W_{\sim i};X_{1i},X_{2i},Y_i|W_i) \notag\\
&\phantom{ww}-\sum_{i=1}^n I(X_{1,i+1}^n,X_{2,i+1}^n,Y_{i+1}^n,W_{\sim i};X_{1i},X_{2i},Y_i|W_i) \notag\\
&\phantom{ww}-\sum_{i=1}^n I(\tilde{Y}_i;\tilde{X}_{1i},\tilde{X}_{2i}|\tilde Z_i)-\sum_{i=1}^n H(Y_i|X_{1i},X_{2i},W_i)-ng(\epsilon) \notag\\
&\stackrel{(d)}\geq \sum_{i=1}^n I(K_1,K_2,\tilde{Y}^n,W_{\sim i},X_{1,i+1}^n,X_{2,i+1}^n;X_{1i},X_{2i},Y_i|W_i) \notag\\
&\phantom{ww}-\sum_{i=1}^n I(\tilde{Y}_i;\tilde{X}_{1i},\tilde{X}_{2i}|\tilde Z_i)-\sum_{i=1}^n H(Y_i|X_{1i},X_{2i},W_i)-2ng(\epsilon) \notag\\
&= \sum_{i=1}^n I(U_{1i},U_{2i};X_{1i},X_{2i},Y_i|W_i) \notag\\
&\phantom{ww}-\sum_{i=1}^n I(\tilde{Y}_i;\tilde{X}_{1i},\tilde{X}_{2i}|\tilde Z_i)-\sum_{i=1}^n H(Y_i|X_{1i},X_{2i},W_i)-2ng(\epsilon) \notag\\
&= nI(U_{1T},U_{2T};X_{1T},X_{2T},Y_T|W_T,T) \notag\\
&\phantom{ww}-nI(\tilde{Y}_T;\tilde{X}_{1T},\tilde{X}_{2T}|\tilde Z_T,T)-n H(Y_T|X_{1T},X_{2T},W_T,T)-2ng(\epsilon) \notag\\
&= nI(U_{1},U_{2};X_{1},X_{2},Y|W,T) \notag\\
&\phantom{ww}-nI(\tilde{X}_{1},\tilde{X}_{2};\tilde{Y}|\tilde Z,T)-n H(Y|X_{1},X_{2},W,T)-2ng(\epsilon),
\end{align}
where (a) follows from the strong secrecy constraint \eqref{eq:refe0inoisy}, (b) follows since $(X_1^n,X_2^n,W^n) \fooA (\tilde{X}_1^n,\tilde{X}_2^n) \fooA (\tilde{Y}^n,\tilde Z^n)$ is a Markov chain and the data processing inequality,
(c) follows from the memoryless nature of the channel $p(\tilde{y},\tilde z|\tilde{x}_1,\tilde{x}_2)$, while (d) follows from \eqref{eqn:totalSRnoisy} and Lemma~\ref{lem:c1}. The individual bounds on $R_{01}$ and $R_{02}$ follow in a similar manner.

\section{Converse Proof of Theorem \ref{thm:indepnoisy}}\label{proof:conv-thm5}
Consider a coding scheme that induces a joint distribution on $(X_1^n,X_2^n,W^n,Y^n,\tilde Z^n)$ which satisfies
\begin{gather}
\lVert p_{X_1^n,X_2^n,W^n,Y^n,\tilde Z^n}-p_{\tilde Z^n} \cdot q^{(n)}_{X_1X_2WY}\rVert_1 \leq \epsilon,\label{eqn:totalSR1noisy}
\end{gather}
for $\epsilon\in(0,\frac{1}{4}]$. Let us now prove the first inequality in Theorem~\ref{thm:indepnoisy}. Recall that $\tilde Y=(f_1(\tilde X_1),f_2(\tilde X_2)) \triangleq (\tilde Y_1,\tilde Y_2)$. 
We have following the chain of inequalities:
\begin{align}
0 &\leq H(\tilde Y_1^n)-I(\tilde Y_1^n; K_1,X_1^n,W^n) \notag\\
&\leq H(\tilde Y_1^n)-I(\tilde Y_1^n; X_1^n|K_1,W^n) \notag\\
&\leq \sum_{i=1}^n H(\tilde Y_{1i})-\sum_{i=1}^n I(\tilde Y_1^n; X_{1i}|K_1,X_{1,i+1}^n,W^n) \notag\\
&= \sum_{i=1}^n H(\tilde Y_{1i})-\sum_{i=1}^n I(\tilde Y_1^n; X_{1i}|K_1,X_{1,i+1}^n,W_{\sim i},W_i) \notag\\
&\stackrel{(a)}= \sum_{i=1}^n H(\tilde Y_{1i})-\sum_{i=1}^n I(K_1,\tilde Y_1^n,X_{1,i+1}^n,W_{\sim i}; X_{1i}|W_i) \notag\\
&\stackrel{(b)}= \sum_{i=1}^n H(\tilde Y_{1i})-\sum_{i=1}^n I(U_{1i}; X_{1i}|W_i) \notag\\
&= nH(\tilde Y_{1T}|T)-nI(U_{1T}; X_{1T}|W_T,T) \notag\\
&= nH(\tilde Y_{1}|T)-nI(U_{1}; X_{1}|W,T), \label{refnoisy1}
\end{align}
where (a) follows from the joint i.i.d. nature of $(X_{1i},W_i),i=1,\dots,n$ and the independence of $K_1$ from $(X_1^n,W^n)$, while (b) follows from an auxiliary random variable identification $U_{1i}=(K_1,X_{1,i+1}^n,\tilde Y_{1}^n,W_{\sim i})$.
The inequality $H(\tilde{Y}_{2}|T) \geq I(U_{2};X_2|W,T)$ follows analogously, with an auxiliary random variable choice given by $U_{2i}=(K_2,X_{2,i+1}^n,\tilde Y_{2}^n)$. 

The lower bound on the shared randomness rate $R_{01}$ is proved next. Consider the sequence of inequalities below:
\begin{align}
nR_{01} &= H(K_1) \notag\\
&\geq H(K_1|\tilde Y_1^n,X_2^n,W^n) \notag\\
&\geq I(K_1;X_1^n,Y^n|\tilde Y_1^n,X_2^n,W^n) \notag\\
&= I(K_1,\tilde Y_1^n;X_1^n,Y^n|X_2^n,W^n)-I(\tilde Y_1^n;X_1^n,Y^n|X_2^n,W^n) \notag\\
&= I(K_1,\tilde Y_1^n;X_1^n,Y^n|X_2^n,W^n)-I(\tilde Y_1^n;X_1^n,Y^n|X_2^n,W^n) \notag\\
&\phantom{www}+I(\tilde Z^n;X_1^n,Y^n|X_2^n,W^n)-I(\tilde Z^n;X_1^n,Y^n|X_2^n,W^n) \notag\\
&\stackrel{(a)}\geq I(K_1,\tilde Y_1^n;X_1^n,Y^n|X_2^n,W^n)-ng(\epsilon)-I(\tilde Y_1^n;X_1^n,Y^n|\tilde Z^n,X_2^n,W^n) \notag\\
&\geq I(K_1,\tilde Y_1^n;X_1^n,Y^n|X_2^n,W^n)-H(\tilde Y_1^n|\tilde Z^n)-ng(\epsilon) \notag\\
&\geq \sum_{i=1}^n I(K_1,\tilde Y_1^n;X_{1i},Y_i|X_{1,i+1}^n,Y_{i+1}^n,X_2^n,W^n) -\sum_{i=1}^n H(\tilde Y_{1i}|\tilde Z_i)-ng(\epsilon) \notag\\
&= \sum_{i=1}^n I(K_1,\tilde Y_1^n,X_{1,i+1}^n,Y_{i+1}^n,X_{2 \sim i}, W_{\sim i};X_{1i},Y_i|X_{2i},W_i) \notag\\
&\phantom{ww}-\sum_{i=1}^n I(X_{1,i+1}^n,Y_{i+1}^n,X_{2 \sim i},W_{\sim i};X_{1i},Y_i|X_{2i},W_i) -\sum_{i=1}^n H(\tilde Y_{1i}|\tilde Z_i)-ng(\epsilon) \notag\\
&\stackrel{(b)}\geq \sum_{i=1}^n I(K_1,\tilde Y_1^n,X_{1,i+1}^n,W_{\sim i};X_{1i},Y_i|X_{2i},W_i)-2ng(\epsilon)-\sum_{i=1}^n H(\tilde Y_{1i}|\tilde Z_i) \notag\\
&= \sum_{i=1}^n I(U_{1i};X_{1i},Y_i|X_{2i},W_i)-\sum_{i=1}^n H(\tilde Y_{1i}|\tilde Z_i)-2ng(\epsilon) \notag\\
&= nI(U_{1T};X_{1T},Y_T|X_{2T},W_T,T) -n H(\tilde Y_{1T}|\tilde Z_T,T)-2ng(\epsilon) \notag\\
&= nI(U_{1};X_{1},Y|X_{2},W,T) -n H(\tilde Y_{1}|\tilde Z,T)-2ng(\epsilon),
\end{align}
where (a) follows from the strong secrecy constraint \eqref{eq:refe0inoisy}, while (b) follows by \eqref{eqn:totalSR1noisy} and Lemma~\ref{lem:c1}.

All that remains to complete the proof of Theorem~\ref{thm:indepnoisy} is the proof of continuity of the derived converse bound at $\epsilon=0$ (note that $g(0):=0$ through continuous extension of the function $g(\epsilon)$). To establish such continuity, it is necessary to establish cardinality bounds on the auxiliary random variables $(U_{1},U_{2})$ to ensure the compactness of the simplex, as outlined in \cite[Lemma VI.5]{cuff2013distributed} and \cite[Lemma 6]{YassaeeGA15}. The cardinality bound $|\mathcal{T}| \leq 3$ on $T$ follows using Caratheodory's theorem~\cite[Appendix A]{el2011network}. On the other hand, the cardinalities of $U_{1}$ and $U_{2}$ can be restricted to $|\mathcal{U}_{1}| \leq |\mathcal{X}_1||\mathcal{X}_2||\mathcal{W}||\mathcal{Y}|$ and $|\mathcal{U}_{2}| \leq |\mathcal{U}_{1}||\mathcal{X}_1||\mathcal{X}_2||\mathcal{W}||\mathcal{Y}|$ respectively following the perturbation method of \cite{gohari2012evaluation}, as detailed in Appendix~\ref{app:cardbndnoisy}.
Finally, by invoking the continuity properties of total variation distance and mutual information in the probability simplex, similar to the approach in \cite[Lemma VI.5]{cuff2013distributed} and \cite[Lemma 6]{YassaeeGA15}, the converse proof for Theorem~\ref{thm:indepnoisy} is complete.

\section{Proof of Theorem~\ref{thm:encsideIBnoisyK0}} \label{proof:thm4}
We only outline the key differences with respect to the proof of Theorem~\ref{thm:encsideIBnoisy}. We define two protocols as follows.\\
\underline{Random Binning Protocol:} Let the random variables $(U_{1}^n,V_{1}^n,U_{2}^n,V_{2}^n,X_1^n,X_2^n,W^n,\tilde{X}_1^n,\tilde{X}_2^n,\tilde{Y}^n,\tilde Z^n,Y^n)$ be drawn i.i.d. according to the joint distribution
\begin{align}
p(x_1,x_2,w,u_{1},v_{1},u_{2},v_{2},\tilde{x}_1,\tilde{x}_2,\tilde{y},\tilde z,y) &=p(x_1,x_2,w)p(u_{2},v_2|x_2)p(\tilde{x}_2|v_{2})p(u_{1},v_1|x_1,u_2,\tilde x_2) \notag\\
&\phantom{ww}\times p(\tilde{x}_1|v_{1},\tilde x_2)p(\tilde{y},\tilde z|\tilde{x}_1,\tilde{x}_2)p(y|u_{1},u_{2},w)
\end{align}
such that the marginal $p(x_1,x_2,w,y)=q(x_1,x_2,w,y)$ and the random variables $(U_{1},U_{2},X_1,X_2,W,Y)$ are independent of $(V_{1},V_{2},\tilde{X}_1,\tilde{X}_2,\tilde{Y},\tilde Z)$. 
The random binning and Slepian-Wolf decoder definitions follow as in the proof of Theorem~\ref{thm:encsideIBnoisy}.
The random p.m.f. induced by this binning scheme is given as follows:
\begin{align*}
&P(x_1^n,x_2^n,w^n,y^n,u_{1}^n,u_{2}^n,v_{1}^n,v_{2}^n,\tilde{x}_1^n,\tilde{x}_2^n,\tilde{y}^n,\tilde z^n,k_1,f_1,k_2,f_2,\hat{u}_{1}^n,\hat{u}_{2}^n) \notag\\
&= p(x_1^n,x_2^n,w^n) P(k_2,f_2|x_2^n) P(u_{2}^n,v_{2}^n|k_2,f_2,x_2^n)p(\tilde{x}_2^n|v_{2}^n)  \notag\\
&\phantom{ww} \times  P(k_1,f_1|x_1^n,u_2^n,\tilde x_2^n) P(u_{1}^n,v_{1}^n|k_1,f_1,x_1^n,\tilde x_2^n) p(\tilde{x}_1^n|v_{1}^n,\tilde x_2^n) \notag\\
&\phantom{ww} \times p(\tilde{y}^n,\tilde z^n|\tilde{x}_1^n,\tilde{x}_2^n)P^{SW}(\hat{u}_{1}^n,\hat{u}_{2}^n|k_1,f_1,k_2,f_2,w^n,\tilde{y}^n) p(y^n|u_{1}^n,u_{2}^n,w^n). 
\end{align*}

\noindent \underline{Random Coding Protocol:} The main difference compared to the proof of Theorem~\ref{thm:encsideIBnoisy} is that Encoder $1$ now observes $(k_1,f_1,x_1^n,\tilde x_2^n)$, and generates $(u_{1}^n,v_1^n)$ according to the p.m.f. $P(u_{1}^n,v_1^n|k_1,f_1,x_1^n,\tilde x_2^n)$ from the binning protocol above. Further, Encoder $1$ then creates the channel input $\tilde{x}_1^n$ according to the p.m.f. $p(\tilde{x}_1^n|v_{1}^n,\tilde x_2^n)$ for transmission over the channel. 
The induced random p.m.f. from the random coding scheme described is
\begin{align*}
&\hat{P}(x_1^n,x_2^n,w^n,y^n,u_{1}^n,u_{2}^n,v_{1}^n,v_{2}^n,\tilde{x}_1^n,\tilde{x}_2^n,\tilde{y}^n,\tilde z^n,k_1,f_1,k_2,f_2,\hat{u}_{1}^n,\hat{u}_{2}^n) \notag\\
&= p^{\text{U}}(k_2)p^{\text{U}}(f_2)p^{\text{U}}(k_1)p^{\text{U}}(f_1)p(x_1^n,x_2^n,w^n)  \notag\\
&\phantom{ww} \times P(u_{2}^n,v_{2}^n|k_2,f_2,x_2^n) p(\tilde{x}_2^n|v_{2}^n) P(u_{1}^n,v_{1}^n|k_1,f_1,x_1^n,\tilde x_2^n) p(\tilde{x}_1^n|v_{1}^n,\tilde x_2^n)  \notag\\
&\phantom{ww} \times p(\tilde{y}^n,\tilde z^n|\tilde{x}_1^n,\tilde{x}_2^n) P^{SW}(\hat{u}_{1}^n,\hat{u}_{2}^n|k_1,f_1,k_2,f_2,w^n,\tilde{y}^n) p(y^n|\hat{u}_{1}^n,\hat{u}_{2}^n,w^n). 
\end{align*}

We next establish constraints that lead to nearly identical induced p.m.f.'s from both protocols.

\noindent \underline{Analysis of Rate Constraints:}\\
Using the fact that $(k_j,f_j)$ are bin indices of $(u_{j}^n,v_j^n)$ for $j \in \{1,2\}$, we impose the conditions
\begin{align}
R_{01}+\tilde{R}_1 &\leq H(U_{1},V_1|X_1,U_2,\tilde X_2) \notag\\
&=H(U_1|X_1,U_2)+H(V_1|\tilde X_2), \label{eq:cond11noisyK0} \\
R_{02}+\tilde{R}_2 &\leq  H(U_{2},V_2|X_2)=H(U_2|X_2)+H(V_2), \label{eq:cond12noisyK0}
\end{align}
that ensure, by invoking \cite[Theorem 1]{yassaee2014achievability}
\begin{align*}
&P(x_1^n,x_2^n,w^n,u_{1}^n,u_{2}^n,v_{1}^n,v_{2}^n,\tilde{x}_1^n,\tilde{x}_2^n,\tilde{y}^n,\tilde z^n,k_1,f_1,k_2,f_2,\hat{u}_{1}^n,\hat{u}_{2}^n) \notag\\
&\phantom{w} \approx \hat{P}(x_1^n,x_2^n,w^n,u_{1}^n,u_{2}^n,v_{1}^n,v_{2}^n,\tilde{x}_1^n,\tilde{x}_2^n,\tilde{y}^n,\tilde z^n,k_1,f_1,k_2,f_2,\hat{u}_{1}^n,\hat{u}_{2}^n).
\end{align*}

The conditions for the success of the Slepian-Wolf decoder and the elimination of additional shared randomness variables $(F_1,F_2)$, along with the proof of strong coordination of the joint p.m.f. of the random variables $(X_1^n,X_2^n,W^n,Y^n)$ and the secrecy constraint follow similarly to the proof of Theorem~\ref{thm:encsideIBnoisy}.
Finally on eliminating $(\tilde{R}_1,\tilde{R}_2)$ from equations \eqref{eq:cond11noisyK0} -- \eqref{eq:cond12noisyK0}, \eqref{eq:cond21noisy} -- \eqref{eq:cond23noisy} and \eqref{eq:cond31secnoisy} -- \eqref{eq:cond33secnoisy} by the FME procedure, we obtain the rate constraints in Theorem~\ref{thm:encsideIBnoisyK0}.

\section{Converse Proof of Proposition~\ref{prop:1}}\label{sec:exproof}
For the converse, we need to prove that for any p.m.f. 
\begin{align*}
p(x_1,&x_2,u_{1},u_{2},\tilde{x}_1,\tilde{x}_2,\tilde z,y)=p(x_1) p(x_2) p(u_{1}|x_1) p(u_{2}|x_2) p(\tilde x_1) p(\tilde x_2) p(\tilde z|\tilde x_1,\tilde x_2) p(y|u_{1},u_{2})
\end{align*}
such that $$\sum\limits_{u_{1},u_{2},\tilde{x}_1,\tilde{x}_2,\tilde z}p(x_1,x_2,u_{1},u_{2},\tilde{x}_1,\tilde{x}_2,\tilde z,y)=q(x_1,x_2,y),$$ we have $I(U_1;X_1) \geq 1$ and $I(U_2;X_2) \geq 1$ (notice that the third inequality $I(U_1;X_1,Y|X_2)-H(\tilde X_1|\tilde Z) \geq 1$ reduces back to $I(U_1;X_1) \geq 1$ since $Y$ is determined by $(X_1,X_2)$ and $\tilde X_1$ is determined by $\tilde Z$ in this example). The independence between $X_1$ and $X_2$ along with the long Markov chain $U_1 \fooA X_1 \fooA X_2 \fooA U_2$ ensures that $(U_1,X_1)$ is independent of $(U_2,X_2)$. We also have the output Markov chain $Y \fooA (U_1,U_2) \fooA (X_1,X_2)$. Clearly, if $H(X_1|U_1)=H(X_2|U_2)=0$, it follows that $I(U_1;X_1)=H(X_1)=1$ and $I(U_2;X_2)=H(X_2)=1$. We now prove that if either $H(X_1|U_1)>0$ or $H(X_2|U_2)>0$, a contradiction arises.

Let $H(X_2|U_2)>0$ (which means $I(U_2;X_2)=H(X_2)-H(X_2|U_2)=1-H(X_2|U_2)<1$) for the sake of contradiction. 
Hence there exists $u_2$ with $P(U_2=u_2)>0$ such that $p_{X_2|U_2=u_2}$ is supported on $\{0,1\}$. We note that the Markov chain $Y \fooA (X_1,U_2) \fooA X_2$ holds because
\begin{align}
I(Y;X_2|X_1,U_2) &\leq I(Y,U_1;X_2|X_1,U_2)\notag\\
&=\!I(U_1;X_2|X_1,U_2)+I(Y;X_2|U_1,U_2,X_1)\notag\\
&\leq\!  I(U_1,X_1;U_2,X_2)\!+I(Y;X_1,X_2|U_1,U_2)\notag\\
&=0, \label{eq:MCprop}
\end{align}
where the last equality follows because $(U_1,X_1)$ is independent of $(U_2,X_2)$ and the Markov chain $Y \fooA (U_1,U_2) \fooA (X_1,X_2)$ holds.

Let us consider the induced distribution given by $p_{Y|X_1=0,U_2=u_2}$. This is well-defined because $P(X_1=0,U_2=u_2)>0$ as $X_1$ is independent of $U_2$ and $P(X_1=0)>0,P(U_2=u_2)>0$. The fact that 
\begin{align*} 
&P(X_2=0|X_1=0,U_2=u_2)=P(X_2=0|U_2=u_2)>0
\end{align*}
along with the Markov chain $Y \fooA (X_1,U_2) \fooA X_2$ implies  
\begin{align}
&P(Y=0|X_1=0,U_2=u_2) =P(Y=0|X_1=0,U_2=u_2,X_2=0)=0, \label{eq:contra1}
\end{align}
and
\begin{align}
&P(Y=1|X_1=0,U_2=u_2) =P(Y=1|X_1=0,U_2=u_2,X_2=0)=0. \label{eq:contra2}
\end{align}
Likewise, the fact that 
\begin{align*}
&P(X_2=1|X_1=0,U_2=u_2)=P(X_2=1|U_2=u_2)>0
\end{align*}
along with the Markov chain $Y \fooA (X_1,U_2) \fooA X_2$ implies  
\begin{align}
&P(Y=\text{`e'}|X_1=0,U_2=u_2) =P(Y=\text{`e'}|X_1=0,U_2=u_2,X_2=1)=0. \label{eq:contra3} 
\end{align}
From \eqref{eq:contra1}--\eqref{eq:contra3}, we are led to a contradiction since $p_{Y|X_1=0,U_2=u_2}$ has to be a probability distribution.

Similarly, let $H(X_1|U_1)>0$ (which means $I(U_1;X_1)=H(X_1)-H(X_1|U_1)=1-H(X_1|U_1)<1$) for the sake of contradiction. Hence there exists $u_1$ with $P(U_1=u_1)>0$ such that $p_{X_1|U_1=u_1}$ is supported on $\{0,1\}$. We note that the Markov chain $Y \fooA (X_2,U_1) \fooA X_1$ holds because
\begin{align}
I(Y;X_1|X_2,U_1) &\leq I(Y,U_2;X_1|X_2,U_1)\notag\\
&=\!I(U_2;X_1|X_2,U_1)+I(Y;X_1|U_1,U_2,X_2)\notag\\
&\leq\!  I(U_1,X_1;U_2,X_2)\!+I(Y;X_1,X_2|U_1,U_2)\notag\\
&=0,
\end{align}
where the last equality follows because $(U_1,X_1)$ is independent of $(U_2,X_2)$ and the Markov chain $Y \fooA (U_1,U_2) \fooA (X_1,X_2)$ holds.

Let us consider the induced distribution given by $p_{Y|X_2=1,U_1=u_1}$. This is well-defined because $P(X_2=1,U_1=u_1)>0$ as $X_2$ is independent of $U_1$ and $P(X_2=1)>0,P(U_1=u_1)>0$.  The fact that 
\begin{align*}
P(X_1=1|X_2=1,U_1=u_1)=P(X_1=1|U_1=u_1)>0
\end{align*}
along with the Markov chain $Y \fooA (X_2,U_1) \fooA X_1$ implies  
\begin{align}
&P(Y=0|X_2=1,U_1=u_1) =P(Y=0|X_2=1,U_1=u_1,X_1=1)=0, \label{eq:contra4} 
\end{align}
and
\begin{align}
&P(Y=\text{`e'}|X_2=1,U_1=u_1) =P(Y=\text{`e'}|X_2=1,U_1=u_1,X_1=1)=0. \label{eq:contra5}  
\end{align}
Likewise, the fact that 
\begin{align*}
P(X_1=0|X_2=1,U_1=u_1)=P(X_1=0|U_1=u_1)>0
\end{align*}
along with the Markov chain $Y \fooA (X_2,U_1) \fooA X_1$ implies  
\begin{align}
&P(Y=1|X_2=1,U_1=u_1) =P(Y=1|X_2=1,U_1=u_1,X_1=0)=0. \label{eq:contra6} 
\end{align}
From \eqref{eq:contra4}--\eqref{eq:contra6}, we are led to a contradiction since $p_{Y|X_2=1,U_1=u_1}$ has to be a probability distribution.

\section{Conclusion} \label{sec:conc}
We investigated strong coordination over a MAC-WT with secrecy constraints. Inner and outer bounds were derived on the trade-off region of shared randomness rates. Due to the distributed nature of the encoding, there is an associated rate-loss that makes it challenging to obtain matching inner and outer bounds in general. However, we derived a tight characterization for the special case of conditionally independent sources and deterministic legitimate receiver's channel. Cooperation between the encoders in the form of cribbing was found to be a useful resource for secure channel simulation. Other multi-terminal configurations, such as for instance a corresponding broadcast channel coordination setting with secrecy constraints, appear to be interesting for further research.

\appendices
\section{Proof of Remark~\ref{rmk:noiselessMAC}} \label{app:noiselessMAC}
Firstly, let us take $\tilde Z=\tilde Y$ since the eavesdropper sees the same channel as the legitimate receiver. We then convert the MAC $p(\tilde{y}|\tilde{x}_1,\tilde{x}_2)$ into two independent noiseless links of rates $R_1$ and $R_2$. In particular, one can consider $\tilde{Y}$ comprising two independent components $(\tilde{Y}_1,\tilde{Y}_2)$, with $R_j=\max_{p(\tilde{x}_j)} I(\tilde{X}_j;\tilde{Y}_j)$, $j \in \{1,2\}$. Substituting $V_{1}=\tilde{X}_1$ and $V_{2}=\tilde{X}_2$ such that $\tilde{X}_1$ and $\tilde{X}_2$ are independent of $(U_1,U_2,X_1,X_2,W,Y)$, Theorem~\ref{thm:encsideIBnoisy} specializes to the set of constraints
\begin{align*} 
R_1 &\geq  I(U_1;X_1|U_2,W,T)\\
R_2 &\geq  I(U_2;X_2|U_1,W,T)\\
R_1+R_2 &\geq I(U_1,U_2;X_1,X_2|W,T)\\
R_{01} &\geq I(U_1;X_1,X_2,Y|W,T)-I(U_1;U_2|W,T) \\
R_{02} &\geq I(U_2;X_1,X_2,Y|W,T)-I(U_1;U_2|W,T)
\end{align*}
\vspace{-8mm}
\begin{align*}
R_2+R_{01} &\geq I(U_1;X_1,X_2,Y|W,T)+I(U_2;X_2|U_1,W,T)\\
R_1+R_{02} &\geq I(U_2;X_1,X_2,Y|W,T)+I(U_1;X_1|U_2,W,T)
\end{align*}
\vspace{-8mm}
\begin{align*}
R_{01}+R_{02} &\geq I(U_1,U_2;X_1,X_2,Y|W,T),
\end{align*}
for some p.m.f. 
\begin{align*}
p(x_1,&x_2,w,t,u_{1},u_{2},y)=p(x_1,x_2,w)p(t)\biggl(\prod_{j=1}^2 p(u_{j}|x_j,t)\biggr) p(y|u_{1},u_{2},w,t)
\end{align*}
such that $\sum\limits_{u_{1},u_{2}} p(x_1,x_2,w,u_{1},u_{2},y|t) =q(x_1,x_2,w,y) \:\: \forall \:\: t.$
This is nothing but the inner bound of \cite[Theorem 1]{RamachandranOSIZS2024}.

\section{Proof of Auxiliary Cardinality Bounds}\label{app:cardbndnoisy} 
The cardinality bound $|\mathcal{T}| \leq 3$ on $T$ follows using Caratheodory's theorem~\cite[Appendix A]{el2011network}. Clearly, to maintain the terms $I(U_1;X_1|W,T)-H(\tilde{Y}_1|T)$, $I(U_2;X_2|W,T)-H(\tilde{Y}_2|T)$ and $I(U_1;X_1,Y|X_2,W,T)-H(\tilde Y_1|\tilde Z,T)$ involved in the rate constraints and thus preserve the rate region, it is sufficient for the alphabet $\mathcal{T}$ to consist of three elements. Next, to bound the cardinalities of $(U_1,U_2)$, consider the rate region with $|\mathcal{T}|=1$:
\begin{align}
H(\tilde{Y}_1) &\geq I(U_1;X_1|W) \label{eq:T3r1} \\
H(\tilde{Y}_2) &\geq I(U_2;X_2|W) \label{eq:T3r2} \\
R_{01} &\geq I(U_1;X_1,Y|X_2,W)-H(\tilde Y_1|\tilde Z), \label{eq:T3r3}
\end{align}
for some p.m.f. 
\begin{align}\label{p.m.f.structure1}
p(x_1,x_2,w,u_1,u_2,\tilde{x}_1,\tilde{x}_2,\tilde z,y) &=p(w)p(x_1|w)p(x_2|w) \biggl(\prod_{j=1}^2 p(u_{j}|x_j)\biggr) \nonumber\\
&\hspace{24pt}\times p(\tilde x_1)p(\tilde x_2)p(\tilde z|\tilde x_1,\tilde x_2)p(y|u_1,u_2,w)
\end{align}
such that 
\begin{align}
\Big\lVert \sum_{u_1,u_2,\tilde{x}_1,\tilde{x}_2,\tilde z} &p(x_1,x_2,w,u_1,u_2,\tilde{x}_1,\tilde{x}_2,\tilde z,y)- q(x_1,x_2,w,y) \Big\rVert\leq \epsilon. \label{eq:pmfmarginalize}
\end{align}
Denote this region as $\mathcal{C}$. 
The auxiliary cardinality bounds will be obtained using the perturbation argument of \cite{gohari2012evaluation}. Firstly, similar to \cite[Appendix E]{kurri2022multiple}, one can show that
\begin{align}
\mathcal{C} = \textup{Closure}\left(\underset{N_{1},N_{2} \geq 0}{\bigcup} \mathcal{C}^{N_{1},N_{2}}\right),
\end{align}
where the region $\mathcal{C}^{N_1,N_2}$, defined for positive integers $N_1$ and $N_2$, is characterized by \eqref{eq:T3r1}--\eqref{eq:T3r3} with respect to random variables $(X_1,X_2,W,U_1,U_2,\tilde{X}_1,\tilde{X}_2,\tilde Z,Y)$ subject to cardinality bounds $|\mathcal{U}_{1}| \leq N_1$ $\&$ $|\mathcal{U}_{2}| \leq N_2$ and a joint p.m.f. satisfying \eqref{p.m.f.structure1} such that the total variation constraint $\lVert p(x_1,x_2,w,y)- q(x_1,x_2,w,y)\rVert\leq \epsilon$ holds.

We next proceed to demonstrate that if $|\mathcal{U}_{1}| \leq N_1$ and $|\mathcal{U}_{2}| \leq N_2$ for some constants $N_1$ and $N_2$, then it is possible to reduce the auxiliary cardinalities to the ones mentioned.
To achieve this, we can focus on optimizing the weighted sum term  $\lambda_1 I(U_1;X_1|W)+\lambda_2 I(U_2;X_2|W)+\lambda_3 I(U_1;X_1,Y|X_2,W)$ for non-negative reals $(\lambda_1,\lambda_2,\lambda_3)$, and introduce new auxiliary random variables with cardinalities bounded as mentioned, while ensuring that the weighted sum is not increased and \eqref{p.m.f.structure1} is satisfied along with $\lVert p(x_1,x_2,w,y)- q(x_1,x_2,w,y)\rVert\leq \epsilon$.

For a given p.m.f. $p(x_1,x_2,w,u_1,u_2,\tilde{x}_1,\tilde{x}_2,\tilde{z},y)$, define the following perturbation
\begin{align*}
&p_{\epsilon}(x_1,x_2,w,u_1,u_2,\tilde{x}_1,\tilde{x}_2,\tilde z,y) =p(x_1,x_2,w,u_1,u_2,\tilde{x}_1,\tilde{x}_2,\tilde z,y)\left(1+\epsilon \phi(u_1)\right).
\end{align*}
Clearly, one must have $\left(1+\epsilon \phi(u_1)\right) \geq 0$ for all $u_1$ along with $\sum_{u_1}p(u_1)\phi(u_1)=0$ so that the distribution $p_{\epsilon}(x_1,x_2,w,u_1,u_2,\tilde{x}_1,\tilde{x}_2,\tilde z,y)$ is a valid p.m.f. Moreover, let us restrict our perturbation $\phi(u_1)$ to satisfy the condition
\begin{align} \label{eq:perturbT3}
&\mathbb{E}\bigl[\phi(U_1)|X_1=x_1,X_2=x_2,W=w,Y=y\bigr] \nonumber\\
&\phantom{www}= \sum_{u_1} p(u_1|x_1,x_2,w,y) \phi(u_1) = 0, \:\: \forall \:\: x_1, x_2, w,y.
\end{align}
We note that a non-zero perturbation satisfying \eqref{eq:perturbT3} exists as long as $|\mathcal{U}_{1}| > |\mathcal{X}_1||\mathcal{X}_2||\mathcal{W}||\mathcal{Y}|$, because in this case the null-space of the constraints has rank not exceeding $|\mathcal{X}_1||\mathcal{X}_2||\mathcal{W}||\mathcal{Y}|$. Notice that such a perturbation satisfying \eqref{eq:perturbT3} preserves the joint distribution of $(X_1,X_2,W,Y)$. Indeed, we have
\begin{align}
p_{\epsilon}(x_1,x_2,w,y) &= \sum_{u_1,u_2} p(x_1,x_2,w,u_1,u_2,y)\left(1+\epsilon \phi(u_1)\right) \notag\\
&= p(x_1,x_2,w,y) \{1+\epsilon \sum_{u_1} p(u_1|x_1,x_2,w,y) \phi(u_1) \} \nonumber\\
&= p(x_1,x_2,w,y), \label{eq:pertpreserveQ}
\end{align}
where the final step results from \eqref{eq:perturbT3}.

Next, let us prove that the perturbed distribution $p_{\epsilon}(\cdot)$ also preserves the p.m.f. structure in \eqref{p.m.f.structure1}. For this, consider
\begin{align}
p_{\epsilon}(x_1,x_2,w,u_1,u_2,\tilde{x}_1,\tilde{x}_2,\tilde z,y) &=p(x_1,x_2,w,u_1,u_2,\tilde{x}_1,\tilde{x}_2,\tilde z,y)\left(1+\epsilon \phi(u_1)\right) \notag\\
&= p(x_1)p(w|x_1)\left\{p(u_1|x_1)\left(1+\epsilon \phi(u_1)\right)\right\}p(\tilde x_1)p(x_2|w)\nonumber\\
&\hspace{0.5cm}\times p(u_2|x_2) p(\tilde x_2) p(\tilde z|\tilde x_1,\tilde x_2)p(y|u_1,u_2,w). \label{cardbnd1e}
\end{align}
Summing over $(x_2,w,u_2,\tilde{x}_1,\tilde{x}_2,\tilde z,y)$, we have the marginal
\begin{align}
p_{\epsilon}(x_1,u_1)&=p(x_1)p(u_1|x_1)\left(1+\epsilon \phi(u_1)\right).
\end{align}
This yields $p_{\epsilon}(u_1|x_1)=p(u_1|x_1)\left(1+\epsilon \phi(u_1)\right)$, because $p_{\epsilon}(x_1)=p(x_1)$ by invoking \eqref{eq:pertpreserveQ}. Hence, we may rewrite \eqref{cardbnd1e} as follows
\begin{align}
p_{\epsilon}(x_1,x_2,w,u_1,u_2,\tilde{x}_1,\tilde{x}_2,\tilde z,y) &=p(w)p(x_1|w)p(x_2|w)p_{\epsilon}(u_1|x_1)p(\tilde x_1) \nonumber\\
&\hspace{1cm}\times p(u_2|x_2)p(\tilde x_2) p(\tilde z|\tilde x_1,\tilde x_2)p(y|u_1,u_2,w), \label{eq:pmfuxz}
\end{align}
so that $p_{\epsilon}(\cdot)$ preserves the p.m.f. structure in \eqref{p.m.f.structure1}.

If the distribution $p(x_1,x_2,w,u_1,u_2,\tilde{x}_1,\tilde{x}_2,\tilde z,y)$ minimizes the objective function $\lambda_1 I(U_1;X_1|W)+\lambda_2 I(U_2;X_2|W)+\lambda_3 I(U_1;X_1,Y|X_2,W)$, then any valid perturbation $p_{\epsilon}(x_1,x_2,w,u_1,u_2,\tilde{x}_1,\tilde{x}_2,\tilde z,y)$ must satisfy the following first derivative extremality condition:
\begin{align}
&\frac{d}{d \epsilon} \biggl(\lambda_1 I_{\epsilon}(U_1;X_1|W)+\lambda_2 I_{\epsilon}(U_2;X_2|W) +\lambda_3 I_{\epsilon}(U_1;X_1,Y|X_2,W)\biggr) \Big|_{\epsilon=0} = 0, \label{eq:extrempert}
\end{align}
wherein $I_{\epsilon}(\cdot)$ is used to emphasize the fact that these mutual information terms are evaluated under $p_{\epsilon}(\cdot)$.
Let us consider the weighted sum term under $p_{\epsilon}(\cdot)$.
\begin{align}
&\lambda_1 I_{\epsilon}(U_1;X_1|W)+\lambda_2 I_{\epsilon}(U_2;X_2|W)+\lambda_3 I_{\epsilon}(U_1;X_1,Y|X_2,W) \notag\\
&\stackrel{(a)}= \lambda_1\left(H(X_1,W)-H(W)+H_{\epsilon}(U_1,W)-H_{\epsilon}(U_1,X_1,W)\!\right)\nonumber\\
&\phantom{w}+\lambda_2\left(H(X_2,W)-H(W)+H_{\epsilon}(U_2,W)-H_{\epsilon}(U_2,X_2,W)\!\right) \notag\\
&\phantom{w}+\lambda_3 \Bigl(H(X_1,X_2,W,Y)-H(X_2,W)+H_{\epsilon}(U_1,X_2,W)-H_{\epsilon}(U_1,X_1,X_2,W,Y)\Bigr) \notag\\
&\stackrel{(b)}= \lambda_1\Big(H(X_1,W)-H(W)+H(U_1,W)+\epsilon H_{\phi}(U_1,W)-H(U_1,X_1,W)-\epsilon H_{\phi}(U_1,X_1,W)\Big)\nonumber\\
&\phantom{w}+\lambda_2\left(H(X_2,W)-H(W)+H_{\epsilon}(U_2,W)-H_{\epsilon}(U_2,X_2,W)\!\right) \notag\\
&\phantom{w}+\lambda_3\Big(H(X_1,X_2,W,Y)-H(X_2,W)+H(U_1,X_2,W)+\epsilon H_{\phi}(U_1,X_2,W)\nonumber\\
&\phantom{www}-H(U_1,X_1,X_2,W,Y)\!-\!\epsilon H_{\phi}(U_1,X_1,X_2,W,Y)\Big)\!, \label{eq:wsumT3}
\end{align}
where (a) follows from \eqref{eq:pertpreserveQ} since the joint p.m.f. of $(X_1,X_2,W,Y)$ is preserved, while (b) results by defining
\begin{align*}
&H_{\phi}(U_1,X_1,X_2,W,Y) = -\sum_{u_1,x_1,x_2,w,y} p(u_1,x_1,x_2,w,y) \phi(u_1) \log p(u_1,x_1,x_2,w,y),
\end{align*}
and likewise for the other $H_{\phi}(\cdot)$ terms.

Now, summing over $(x_1,u_1,\tilde{x}_1,\tilde{x}_2,\tilde z,y)$ in \eqref{eq:pmfuxz},
\begin{align}
p_{\epsilon}(x_2,u_2,w) = p(x_2,u_2,w).
\end{align}
This means that the terms $H_{\epsilon}(U_2,X_2,W)$ and $H_{\epsilon}(U_2,W)$ are preserved under the perturbation $p_{\epsilon}(\cdot)$.
Therefore, we can rewrite the weighted sum in \eqref{eq:wsumT3} as:
\begin{align}
&\lambda_1 I_{\epsilon}(U_1;X_1|W)+\lambda_2 I_{\epsilon}(U_2;X_2|W) +\lambda_3 I_{\epsilon}(U_1;X_1,Y|X_2,W) \notag\\
&= \lambda_1\left(I(U_1;X_1|W)+\epsilon H_{\phi}(U_1,W)-\epsilon H_{\phi}(U_1,X_1,W)\right) \nonumber\\
&\hspace{12pt}+\lambda_2 I(U_2;X_2|W)+\lambda_3 \Bigl(I(U_1;X_1,Y|X_2,W)+\epsilon H_{\phi}(U_1,X_2,W)-\epsilon H_{\phi}(U_1,X_1,X_2,W,Y)\Bigr). \label{eq:wsum2T3}
\end{align}
We invoke the first derivative condition \eqref{eq:extrempert} in \eqref{eq:wsum2T3} to obtain
\begin{align}
&\lambda_1\left(H_{\phi}(U_1,W)-H_{\phi}(U_1,X_1,W)\right) +\lambda_3\left(H_{\phi}(U_1,X_2,W)-H_{\phi}(U_1,X_1,X_2,W,Y)\right) = 0. \label{eq:wsum5T3} 
\end{align}
Finally, \eqref{eq:wsum5T3} and \eqref{eq:wsum2T3} together imply
\begin{align*}
&\lambda_1 I_{\epsilon}(U_1;X_1|W)+\lambda_2 I_{\epsilon}(U_2;X_2|W) +\lambda_3 I_{\epsilon}(U_1;X_1,Y|X_2,W) \notag\\
&\phantom{w}= \lambda_1 I(U_1;X_1|W)+\lambda_2 I(U_2;X_2|W) +\lambda_3 I(U_1;X_1,Y|X_2,W).
\end{align*}
Thus, if $p(x_1,x_2,w,u_1,u_2,\tilde{x}_1,\tilde{x}_2,\tilde z,y)$ achieves the minimum of the weighted sum rate, then this minimum is preserved for any valid perturbation $p_{\epsilon}(x_1,x_2,w,u_1,u_2,\tilde{x}_1,\tilde{x}_2,\tilde z,y)$ that satisfies equation~\eqref{eq:perturbT3}. We can pick $\epsilon$ that satisfies
\begin{align*}
\min_{u_1} \left(1+\epsilon \phi(u_1)\right) = 0,
\end{align*}
with $u_1=u_1^{*}$ being the minimizer above. It is immediate that $p_{\epsilon}(u_1^{*})=0$, so that there exists a $U_1$ with cardinality not exceeding $|\mathcal{U}_{1}|-1$ such that the weighted sum $\lambda_1 I(U_1;X_1|W)+\lambda_2 I(U_2;X_2|W)+\lambda_3 I(U_1;X_1,Y|X_2,W)$ is preserved. In other words, the cardinality of $U_1$ drops by $1$. This may continue until $|\mathcal{U}_{1}|=|\mathcal{X}_1||\mathcal{X}_2||\mathcal{W}||\mathcal{Y}|$, at which point, the existence of a non-zero perturbation $\phi(u_1)$ satisfying \eqref{eq:perturbT3} is no longer guaranteed. Thus, the cardinality of $U_1$ may be limited such that $|\mathcal{U}_{1}| \leq |\mathcal{X}_1||\mathcal{X}_2||\mathcal{W}||\mathcal{Y}|$. 

The bound on $|\mathcal{U}_2|$ can be shown similarly, by perturbing $U_2$. For a given $p(x_1,x_2,w,u_1,u_2,\tilde{x}_1,\tilde{x}_2,\tilde z,y)$, consider another perturbation 
\begin{align*}
&p'_{\epsilon}(x_1,x_2,w,u_1,u_2,\tilde{x}_1,\tilde{x}_2,\tilde z,y) =p(x_1,x_2,w,u_1,u_2,\tilde{x}_1,\tilde{x}_2,\tilde z,y)\left(1+\epsilon \phi'(u_2)\right).
\end{align*}
Furthermore, we consider perturbations $\phi'(u_2)$ such that
\begin{align} \label{eq:perturbT3U2}
&\mathbb{E}\bigl[\phi'(U_2)|U_1=u_1,\!X_1=x_1,\!X_2=x_2,\!W=w,\!Y=y\bigr] \nonumber\\
&\phantom{www}= \sum_{u_2} p(u_2|u_1,\!x_1,\!x_2,\!w,\!y) \phi'(u_2) = 0, \: \forall \: u_1, x_1,x_2,w,y.
\end{align}
Clearly, a non-zero perturbation that obeys \eqref{eq:perturbT3U2} exists provided $|\mathcal{U}_2| > |\mathcal{U}_{1}||\mathcal{X}_1||\mathcal{X}_2||\mathcal{W}||\mathcal{Y}|$.
Similar to the analysis for the perturbation $\phi(u_1)$ earlier, one may verify that $p'_{\epsilon}(\cdot)$  preserves $p(x_1,x_2,w,y)$ as well as the structure of the p.m.f. in \eqref{p.m.f.structure1}. With this, it can be verified similarly that the weighted sum term under the perturbed distribution $p'_{\epsilon}(\cdot)$ is preserved. 
To complete the proof, we choose $\epsilon$ such that $\min_{u_2} \left(1+\epsilon \phi'(u_2)\right) = 0$, which results in reducing the cardinality of $U_2$ by $1$. We can repeat this process iteratively, and thus $|\mathcal{U}_2| \leq |\mathcal{U}_{1}||\mathcal{X}_1||\mathcal{X}_2||\mathcal{W}||\mathcal{Y}|$.

\bibliographystyle{IEEEtran}
\bibliography{mylit}

\begin{thebibliography}{10}
\providecommand{\url}[1]{#1}
\csname url@samestyle\endcsname
\providecommand{\newblock}{\relax}
\providecommand{\bibinfo}[2]{#2}
\providecommand{\BIBentrySTDinterwordspacing}{\spaceskip=0pt\relax}
\providecommand{\BIBentryALTinterwordstretchfactor}{4}
\providecommand{\BIBentryALTinterwordspacing}{\spaceskip=\fontdimen2\font plus
\BIBentryALTinterwordstretchfactor\fontdimen3\font minus
  \fontdimen4\font\relax}
\providecommand{\BIBforeignlanguage}[2]{{%
\expandafter\ifx\csname l@#1\endcsname\relax
\typeout{** WARNING: IEEEtran.bst: No hyphenation pattern has been}%
\typeout{** loaded for the language `#1'. Using the pattern for}%
\typeout{** the default language instead.}%
\else
\language=\csname l@#1\endcsname
\fi
#2}}
\providecommand{\BIBdecl}{\relax}
\BIBdecl

\bibitem{ahlswede1993common}
R.~Ahlswede and I.~Csisz{\'a}r, ``Common randomness in information theory and
  cryptography. {I. S}ecret sharing,'' \emph{IEEE Transactions on Information
  Theory}, vol.~39, no.~4, pp. 1121--1132, 1993.

\bibitem{ahlswede1998common}
------, ``Common randomness in information theory and cryptography. {II. CR}
  capacity,'' \emph{IEEE Transactions on Information Theory}, vol.~44, no.~1,
  pp. 225--240, 1998.

\bibitem{cuff2010coordination}
P.~Cuff, H.~Permuter, and T.~Cover, ``Coordination capacity,'' \emph{IEEE
  Transactions on Information Theory}, vol.~56, no.~9, pp. 4181--4206, 2010.

\bibitem{willems1985discrete}
F.~Willems and E.~Van~der Meulen, ``The discrete memoryless multiple-access
  channel with cribbing encoders,'' \emph{IEEE Transactions on Information
  Theory}, vol.~31, no.~3, pp. 313--327, 1985.

\bibitem{BennettSST02}
C.~{Bennett}, P.~{Shor}, J.~{Smolin}, and A.~{Thapliyal},
  ``Entanglement-assisted capacity of a quantum channel and the reverse
  {S}hannon theorem,'' \emph{IEEE Transactions on Information Theory}, vol.~48,
  no.~10, pp. 2637--2655, 2002.

\bibitem{6757002}
C.~{Bennett}, I.~{Devetak}, A.~{Harrow}, P.~{Shor}, and A.~{Winter}, ``The
  quantum reverse {S}hannon theorem and resource tradeoffs for simulating
  quantum channels,'' \emph{IEEE Transactions on Information Theory}, vol.~60,
  no.~5, pp. 2926--2959, 2014.

\bibitem{cuff2013distributed}
P.~Cuff, ``Distributed channel synthesis,'' \emph{IEEE Transactions on
  Information Theory}, vol.~59, no.~11, pp. 7071--7096, 2013.

\bibitem{YassaeeGA15}
M.~Yassaee, A.~Gohari, and M.~Aref, ``Channel simulation via interactive
  communications,'' \emph{IEEE Transactions on Information Theory}, vol.~61,
  no.~6, pp. 2964--2982, 2015.

\bibitem{SatpathyC16}
S.~{Satpathy} and P.~{Cuff}, ``Secure cascade channel synthesis,'' \emph{IEEE
  Transactions on Information Theory}, vol.~62, no.~11, pp. 6081--6094, 2016.

\bibitem{vellambi2017strong}
B.~Vellambi, J.~Kliewer, and M.~Bloch, ``Strong coordination over multi-hop
  line networks using channel resolvability codebooks,'' \emph{IEEE
  Transactions on Information Theory}, vol.~64, no.~2, pp. 1132--1162, 2017.

\bibitem{ramachandran2020strong}
V.~Ramachandran, S.~R.~B. Pillai, and V.~M. Prabhakaran, ``Strong coordination
  with side information,'' in \emph{2020 IEEE International Symposium on
  Information Theory (ISIT)}, 2020, pp. 1564--1569.

\bibitem{kurri2022multiple}
G.~R. Kurri, V.~Ramachandran, S.~R.~B. Pillai, and V.~M. Prabhakaran,
  ``Multiple access channel simulation,'' \emph{IEEE Transactions on
  Information Theory}, vol.~68, no.~11, pp. 7575--7603, 2022.

\bibitem{RamachandranOSITW2024}
V.~Ramachandran, T.~J. Oechtering, and M.~Skoglund, ``Multi-terminal strong
  coordination with degraded source observations,'' in \emph{2024 IEEE
  Information Theory Workshop (ITW)}, 2024.

\bibitem{HaddadpourYAG13}
F.~Haddadpour, M.~H. Yassaee, M.~R. Aref, and A.~Gohari, ``When is it possible
  to simulate a {DMC} channel from another?'' in \emph{IEEE Information Theory
  Workshop}, 2013, pp. 1--5.

\bibitem{HaddadpourYBGAA17}
F.~Haddadpour, M.~H. Yassaee, S.~Beigi, A.~Gohari, and M.~R. Aref, ``Simulation
  of a channel with another channel,'' \emph{IEEE Transactions on Information
  Theory}, vol.~63, no.~5, pp. 2659--2677, 2017.

\bibitem{CerviaLLB20}
G.~Cervia, L.~Luzzi, M.~Le~Treust, and M.~R. Bloch, ``Strong coordination of
  signals and actions over noisy channels with two-sided state information,''
  \emph{IEEE Transactions on Information Theory}, vol.~66, no.~8, pp.
  4681--4708, 2020.

\bibitem{gohari2012secure}
A.~Gohari, M.~H. Yassaee, and M.~R. Aref, ``Secure channel simulation,'' in
  \emph{2012 IEEE Information Theory Workshop}.\hskip 1em plus 0.5em minus
  0.4em\relax IEEE, 2012, pp. 406--410.

\bibitem{data2020interactive}
D.~Data, G.~R. Kurri, J.~Ravi, and V.~M. Prabhakaran, ``Interactive secure
  function computation,'' \emph{IEEE Transactions on Information Theory},
  vol.~66, no.~9, pp. 5492--5521, 2020.

\bibitem{cervia2020secure}
G.~Cervia, G.~Bassi, and M.~Skoglund, ``Secure strong coordination,'' in
  \emph{2020 IEEE Conference on Communications and Network Security
  (CNS)}.\hskip 1em plus 0.5em minus 0.4em\relax IEEE, 2020, pp. 1--6.

\bibitem{RamachandranOSISIT2024}
V.~Ramachandran, T.~J. Oechtering, and M.~Skoglund, ``Multi-terminal strong
  coordination over noisy channels with secrecy constraints,'' in \emph{2024
  IEEE International Symposium on Information Theory (ISIT)}, 2024.

\bibitem{RamachandranOSIZS2024}
------, ``Multi-terminal strong coordination over noiseless networks with
  secrecy constraints,'' in \emph{2024 International Zurich Seminar on
  Information and Communication (IZS)}.\hskip 1em plus 0.5em minus 0.4em\relax
  ETH Zurich Library, 2024.

\bibitem{berger1977multiterminal}
T.~Berger, ``Multiterminal source coding,'' \emph{The information theory
  approach to communications}, vol. 229, pp. 171--231, 1977.

\bibitem{SefidgaranT11}
M.~{Sefidgaran} and A.~{Tchamkerten}, ``Computing a function of correlated
  sources: A rate region,'' in \emph{IEEE International Symposium on
  Information Theory}, 2011, pp. 1856--1860.

\bibitem{SefidgaranT16}
------, ``Distributed function computation over a rooted directed tree,''
  \emph{IEEE Transactions on Information Theory}, vol.~62, no.~12, pp.
  7135--7152, 2016.

\bibitem{Gastpar04}
M.~{Gastpar}, ``The {W}yner-{Z}iv problem with multiple sources,'' \emph{IEEE
  Transactions on Information Theory}, vol.~50, no.~11, pp. 2762--2768, 2004.

\bibitem{KornerM79}
J.~{K\"{o}rner} and K.~{Marton}, ``How to encode the modulo-two sum of binary
  sources (corresp.),'' \emph{IEEE Transactions on Information Theory},
  vol.~25, no.~2, pp. 219--221, 1979.

\bibitem{SefidgaranA11}
M.~Sefidgaran and A.~Tchamkerten, ``On computing a function of correlated
  sources,'' \emph{arXiv preprint arXiv:1107.5806}, 2011.

\bibitem{kaspi1982rate}
A.~Kaspi and T.~Berger, ``Rate-distortion for correlated sources with partially
  separated encoders,'' \emph{IEEE Transactions on Information Theory},
  vol.~28, no.~6, pp. 828--840, 1982.

\bibitem{yassaee2014achievability}
M.~Yassaee, M.~Aref, and A.~Gohari, ``Achievability proof via output statistics
  of random binning,'' \emph{IEEE Transactions on Information Theory}, vol.~60,
  no.~11, pp. 6760--6786, 2014.

\bibitem{slepian1973noiseless}
D.~Slepian and J.~Wolf, ``Noiseless coding of correlated information sources,''
  \emph{IEEE Transactions on information Theory}, vol.~19, no.~4, pp. 471--480,
  1973.

\bibitem{pierrot2013joint}
A.~J. Pierrot and M.~R. Bloch, ``Joint channel intrinsic randomness and channel
  resolvability,'' in \emph{2013 IEEE Information Theory Workshop (ITW)}.\hskip
  1em plus 0.5em minus 0.4em\relax IEEE, 2013, pp. 1--5.

\bibitem{cervia2018coordination}
G.~Cervia, ``Coordination of autonomous devices over noisy channels: capacity
  results and coding techniques,'' Ph.D. dissertation, Universit{\'e} de Cergy
  Pontoise, 2018.

\bibitem{cover2012elements}
T.~M. Cover and J.~A. Thomas, \emph{Elements of information theory}.\hskip 1em
  plus 0.5em minus 0.4em\relax John Wiley \& Sons, 2012.

\bibitem{el2011network}
A.~El~Gamal and Y.-H. Kim, \emph{Network information theory}.\hskip 1em plus
  0.5em minus 0.4em\relax Cambridge University Press, 2011.

\bibitem{gohari2012evaluation}
A.~A. Gohari and V.~Anantharam, ``Evaluation of {M}arton's inner bound for the
  general broadcast channel,'' \emph{IEEE Transactions on Information Theory},
  vol.~58, no.~2, pp. 608--619, 2012.

\bibitem{ramachandran2018communication}
V.~Ramachandran, S.~R.~B. Pillai, and V.~M. Prabhakaran, ``Communication and
  state estimation over a state-dependent {G}aussian multiple-access channel,''
  in \emph{2018 Twenty Fourth National Conference on Communications (NCC)}, pp.
  1--6.

\bibitem{ramachandran2017feedback}
V.~Ramachandran and S.~R.~B. Pillai, ``Feedback-capacity of degraded {G}aussian
  vector {BC} using directed information and concave envelopes,'' in \emph{2017
  Twenty-third National Conference on Communications (NCC)}.\hskip 1em plus
  0.5em minus 0.4em\relax IEEE, 2017, pp. 1--6.

\end{thebibliography}
\nocite{ramachandran2018communication}
\nocite{ramachandran2017feedback}
\end{document}